%% file: mt_kahypar.tex
\documentclass[twoside,leqno,twocolumn]{article}

\usepackage[letterpaper]{geometry}

\usepackage{hyperref}
\usepackage{ltexpprt}
\usepackage{color}
\usepackage{bibentry}
\usepackage{amsmath,amsfonts,amssymb,pifont,marvosym}
\usepackage{wasysym}
\usepackage{biblatex}
\usepackage{multirow}
\usepackage{booktabs}
\usepackage{longtable}
\usepackage{tabularx}
\usepackage{graphicx}
\usepackage{subcaption}
\usepackage[breakall]{truncate}
\usepackage{datatool}
\usepackage{xifthen}
\usepackage{lipsum}
\DTLsetseparator{ = }

\usepackage{tikz}
\usetikzlibrary{calc}

\usepackage{pgfplots}
\usepgfplotslibrary{external}
\tikzsetexternalprefix{experiments/}

\input{macros}

\begin{document}

\title{
	\Large Scalable Shared-Memory Hypergraph Partitioning
	\thanks{This work was partially supported by DFG grants WA654/19-2, SA933/10-2 and SA933/11-1. The authors acknowledge support by the state of Baden-Württemberg through bwHPC. The authors also thank Michael Hamann for fruitful discussions and ideas.}
}

\author{
Lars Gottesbüren\thanks{Karlsruhe Institute of Technology, Karlsruhe, Germany. \{lars.gottesbueren, tobias.heuer, sanders\}@kit.edu, research@sebastianschlag.de}
\and Tobias Heuer\footnotemark[2]
\and Peter Sanders\footnotemark[2]
\and Sebastian Schlag\footnotemark[2]
}
\date{}

\maketitle







\begin{abstract} \small\baselineskip=9pt

Hypergraph partitioning is an important preprocessing step for optimizing data placement and minimizing communication volumes in high-performance computing applications.
To cope with ever growing problem sizes, it has become increasingly important to develop fast parallel partitioning algorithms whose solution quality is competitive with existing sequential algorithms.

To this end, we present \mtkahypar, the first shared-memory multilevel hypergraph partitioner with
parallel implementations of many techniques used by the sequential, high-quality partitioning systems:
a parallel coarsening algorithm that uses parallel community detection as guidance,
initial partitioning via parallel recursive bipartitioning with work-stealing,
a scalable label propagation refinement algorithm, and the first fully-parallel direct $k$-way formulation of the classical FM algorithm.

Experiments performed on a large benchmark set of instances from various application domains demonstrate the scalability and effectiveness of our approach. With 64 cores, we observe self-relative speedups of up to 51 and a harmonic mean speedup of 23.5.
In terms of solution quality, we outperform the distributed hypergraph partitioner \texttt{Zoltan} on 95\% of the instances while also being a factor of 2.1 faster.
With just four cores, \mtkahypar~is also slightly faster than the fastest sequential multilevel partitioner \texttt{PaToH} while producing better solutions on 83\% of all instances. The sequential high-quality partitioner \texttt{KaHyPar} still finds better solutions than our parallel approach, especially when using max-flow-based refinement. This, however, comes at the cost of considerably longer running times.

\end{abstract}

\section{Introduction}

Balanced $k$-way hypergraph partitioning (HGP) is a classical, well-studied optimization problem, whose goal is to divide the vertices of a hypergraph into $k$ blocks of bounded size while minimizing an objective function on hyperedges that connect two or more blocks.
A commonly used objective function is the \emph{connectivity metric}, which aims to minimize the number of blocks connected by a hyperedge.
HGP has applications in various domains such as VLSI design~\cite{ALPERT-SURVEY}, scientific computing~\cite{PATOH}, storage sharding for distributed databases~\cite{schism, SHP}, and hypergraph processing frameworks for network analysis~\cite{hg-processing-framework-MESH, hg-processing-hyperx}.
Generally speaking, it can be used to accurately optimize data or workload distribution for distributed computer systems with limited memory and expensive communication.
Since hypergraph partitioning is NP-hard~\cite{LENGAUER}, most state-of-the-art partitioners use heuristic \emph{multilevel} algorithms~\cite{HENDRICKSON}.

These algorithms create a hierarchy of successively smaller and structurally similar hypergraphs by \emph{contracting} pairs or clusters of vertices (\emph{coarsening} phase).
Once the coarsest hypergraph is small enough, an \emph{initial partition} is computed.
Subsequently, the contractions are reverted level-by-level, and, on each level, \emph{local search} heuristics are used to improve the partition from the previous level (\emph{refinement phase}), e.g. using the classical FM algorithm~\cite{FM} or simple greedy heuristics~\cite{HMETIS-K}.

In recent years, problem sizes have grown beyond what full-fledged multilevel single-threaded codes can handle in reasonable time.
This resulted in very fast and simple algorithms that avoid the multilevel framework but yield solutions of inferior quality~\cite{TheHypeIsOver}.
As solution quality directly impacts the scalability of the application this seems to be a bad trade-off.
Hence, parallelization of the components of the multilevel framework is an important avenue for research.
Most of the work on parallelization has focused on distributed-memory systems~\cite{ZOLTAN, PARKWAY-2}.
However, shared-memory systems offer more algorithmic options, so that solution quality will be closer to the highly-tuned sequential codes.

\paragraph{Contribution.}

In this work, we present our shared-memory parallel multilevel hypergraph partitioning system \texttt{Mt-KaHyPar} that is able to partition hypergraphs with billions of pins in a matter of minutes.
To the best of our knowledge, it is the first shared-memory parallel hypergraph partitioner, and the first partitioner to implement fully-parallel direct k-way FM refinement~\cite{FM}.
We contribute parallel implementations of a clustering algorithm for coarsening, an efficient parallel contraction algorithm, a community detection algorithm to guide the coarsening, initial partitioning via recursive bipartitioning with work-stealing, a number of techniques to maintain and verify gains in parallel, as well as a scalable label propagation refinement algorithm~\cite{HMETIS-K}.
Our gain calculation techniques include an exact gain recalculation algorithm that allows for parallelizing a step of the FM algorithm that was previously deemed inherently sequential~\cite{MT-METIS-REFINEMENT}.
The algorithms are implemented using the parallel primitives of Intel's TBB library to leverage work stealing in all stages.

\paragraph{Results.}
\texttt{Mt-KaHyPar} achieves good speedups -- up to a factor of 51 on a machine with 64 cores and a harmonic mean speedup of $23.5$.
It provides similar quality as the fastest sequential partitioner \texttt{PaToH}~\cite{PATOH} in its quality preset and has similar average running time as \texttt{PaToH}'s default preset when using 10 cores of a 20-core machine.
Further, it is much faster than the high-quality sequential partitioners \hmetis~and \kahypar~and produces solutions of similar quality when given the same running time for additional repetitions.
On a second benchmark set with larger instances, it produces better solutions than the distributed hypergraph partitioner \zoltan~on 95\% of the instances, while still being a factor of 2.1 faster, when both use 64 cores.
It is even slightly faster than \texttt{PaToH}'s default preset with just 4 cores on these larger instances and produces partitions with higher quality on 83\% of the instances.

\paragraph{Outline.}
Our paper is structured according to the different phases of the multilevel framework.
We introduce notation in Section~\ref{sec:preliminaries} and discuss related work in Section~\ref{sec:related_work}.
Subsequently, we describe our coarsening algorithm in Section~\ref{sec:coarsening}, initial partitioning in Section~\ref{sec:initial_partitioning}, various gain calculation techniques in Section~\ref{sec:gains}, as well as our refinement algorithms in Sections~\ref{sec:label-propagation} - \ref{sec:rebalancing}.
Additionally, we describe engineering aspects in Section~\ref{sec:details}, before discussing our extensive experimental evaluation in Section~\ref{sec:experiments} and presenting future work in Section~\ref{sec:conclusion}.

\section{Preliminaries}\label{sec:preliminaries}

A \emph{weighted hypergraph} $H=(V,E,c,\omega)$ is defined as a set of vertices $V$ and a set of hyperedges/nets $E$ with vertex weights $c:V \to \mathbb{R}_{>0}$ and net weights $\omega:E \to \mathbb{R}_{>0}$, where each net $e$ is a subset of the vertex set $V$ (i.e., $e \subseteq V$).
The vertices of a net are called its \emph{pins}.
We extend $c$ and $\omega$ to sets in the natural way, i.e., $c(U) :=\sum_{v\in U} c(v)$ and $\omega(F) :=\sum_{e \in F} \omega(e)$.
A vertex $v$ is \emph{incident} to a net $e$ if $v \in e$.
$\mathrm{I}(v)$ denotes the set of all incident nets of $v$.
The \emph{degree} of a vertex $v$ is $d(v) := |\mathrm{I}(v)|$.
The \emph{size} $|e|$ of a net $e$ is the number of its pins.
We call two nets $e_i$ and $e_j$ with different identifiers $i \neq j$ \emph{identical} if they have the same pins $e_i = e_j$.

A \emph{$k$-way partition} of a hypergraph $H$ is a surjective function $\Partition : V \to \{1, \dots, k\}$.
The blocks $V_i := \Partition^{-1}(i)$ of $\Partition$ are the inverse images.
A $2$-way partition is also called a \emph{bipartition}.
We call $\Partition$ \emph{$\varepsilon$-balanced} if each block $V_i$ satisfies the \emph{balance constraint}: $c(V_i) \leq L_{\max} := (1+\varepsilon)\lceil \frac{c(V)}{k} \rceil$ for some parameter $\mathrm{\varepsilon} \in (0,1)$.
The \emph{balanced $k$-way hypergraph partitioning problem} (HGP) asks for an $\varepsilon$-balanced $k$-way partition of $H$.

For each net $e$, $\conset(e) := \{V_i \mid  V_i \cap e \neq \emptyset\}$ denotes the \emph{connectivity set} of $e$.
The \emph{connectivity} $\con(e)$ of a net $e$ is the cardinality of its connectivity set, i.e., $\con(e) := |\conset(e)|$.
A net is called a \emph{cut net} if $\con(e) > 1$, otherwise  (i.e., if $|\mathrm{\lambda}(e)|=1$) it is called an \emph{internal} net.
A vertex $u$ that is incident to at least one cut net is called a \emph{boundary vertex}.
The number of pins of a net $e$ in block $V_i$ is denoted by $\pinsinpart(e,V_i) := |\{V_i \cap e \}|$.

In this paper we focus on the \emph{connectivity metric} $(\lambda - 1)(\Pi) := \sum_{e \in E} (\lambda(e) - 1) \: \omega(e)$.
Given a partition $\Partition$, moving $u$ from its block $\Partition(u)$ to $V_j$ improves the connectivity metric by $\omega(\{ e \in I(u) \mid \pinsinpart(e, \Partition(u)) = 1 \}) - \omega(\{e \in I(u) \mid \pinsinpart(e, V_j) = 0 \})$.
This term is called the \emph{gain} of the move.

\section{Related Work}
\label{sec:related_work}

\input{sections/related_work}

\section{Coarsening}\label{sec:coarsening}
\input{sections/coarsening}

\section{Initial Partitioning}\label{sec:initial_partitioning}
\input{sections/initial_partitioning}

\input{sections/refinement}

\section{Engineering Aspects}\label{sec:details}
\input{sections/details}

\section{Experiments}\label{sec:experiments}

\input{sections/experiments}

\section{Conclusion and Future Work}\label{sec:conclusion}

We propose a shared-memory multilevel hypergraph partitioning algorithm that guarantees balanced solutions and exhibits good speedups and solution quality, even compared with sequential algorithms.
This is achieved by carefully parallelizing and engineering each phase of the multilevel framework.
To further improve solution quality, future work includes parallelizing flow-based refinement techniques~\cite{REBAHFC, KAHYPAR-HFC, KAHYPAR-MF} as well as KaHyPar's n-level hierarchy approach~\cite{KAHYPAR-K, KaHyPar-R}.
Additionally, we want to speed up our FM implementation, and improve the solution quality of label propagation, e.g., using hill-scanning~\cite{MT-METIS-REFINEMENT} or caching negative gain moves to enable otherwise infeasible high-gain moves.
Finally, our parallel gain recalculation can replace the sequential recalculation in the parallel graph partitioner \texttt{Mt-KaHIP}.


\printbibliography

\appendix

\onecolumn

%



\section{Absolute Running Times}\label{sec:appendix:absolute_times}

\begin{figure*}[!ht]
  \centering
  \begin{minipage}{.95\textwidth}
    \centering
    \externalizedfigure{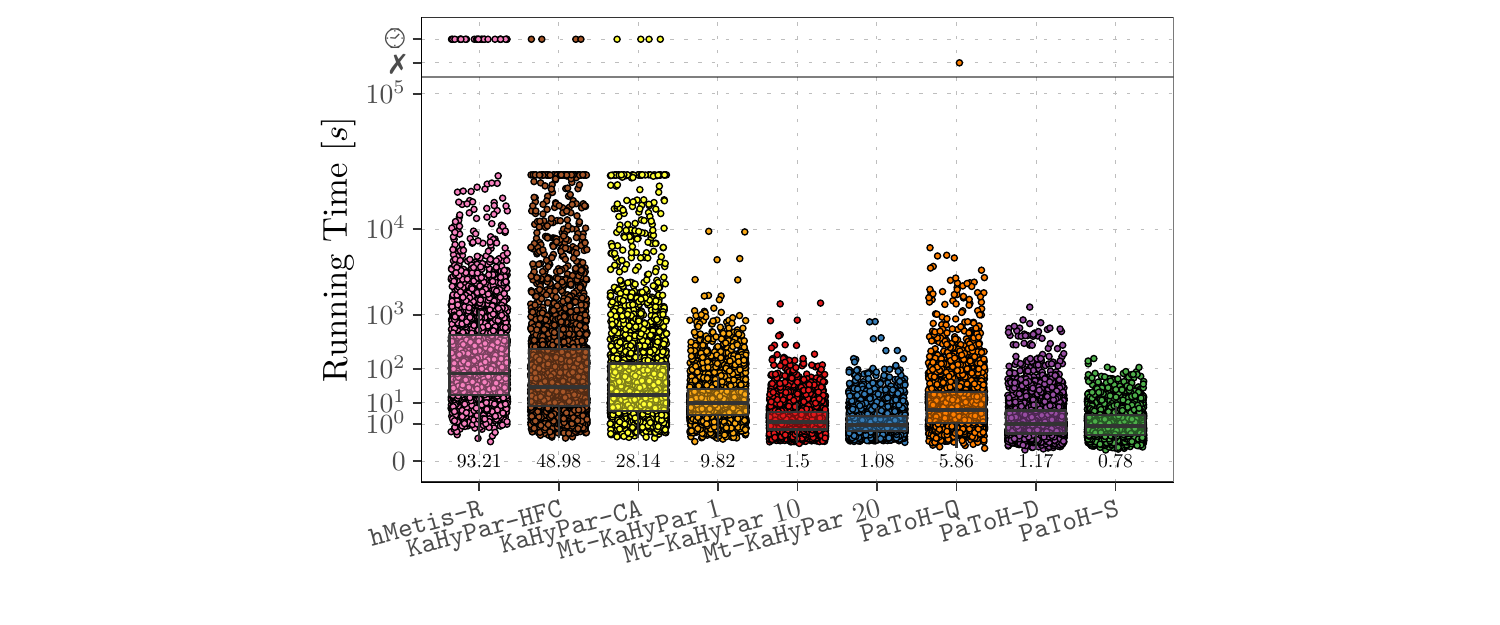}{experiments/kahypar_benchmark_set/kahypar_runtime.tex} %
  \end{minipage} %
  \begin{minipage}{.49\textwidth}
    \externalizedfigure{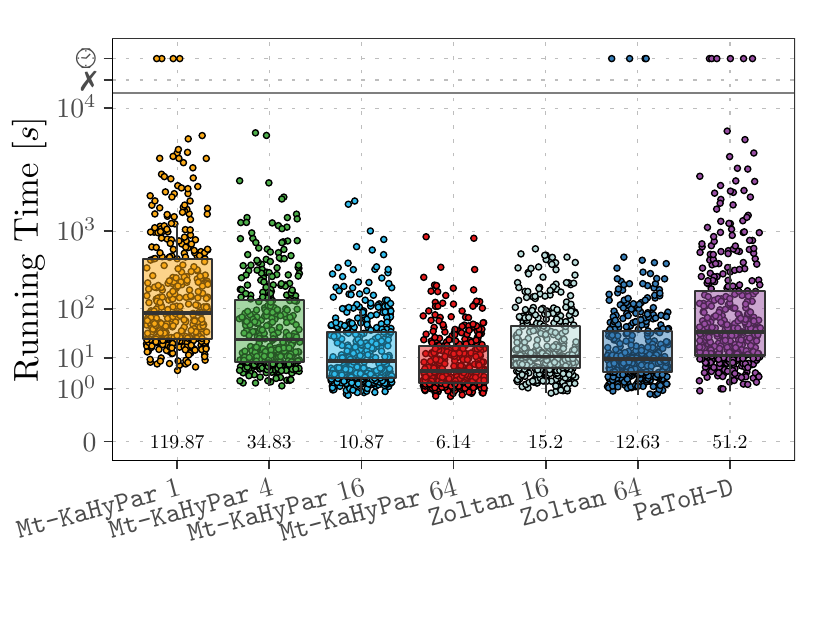}{experiments/big_benchmark_set/big_runtime.tex} %
  \end{minipage} %
  \begin{minipage}{.49\textwidth}
    \externalizedfigure{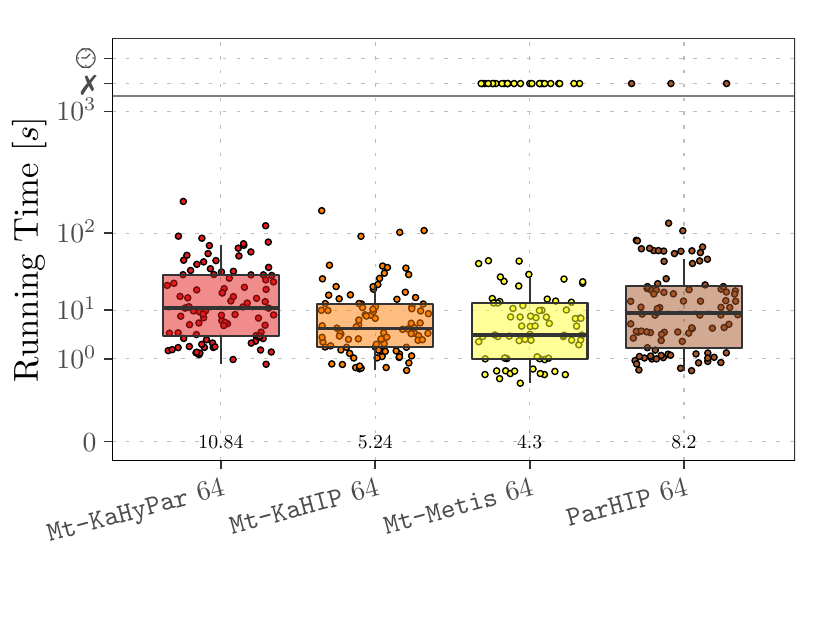}{experiments/graph/graph_runtime.tex} %
  \end{minipage} %
  \vspace{-0.5cm}
  \caption{
	  	Absolute running times of all partitioners used in our experiments in a combined box-and-scatter-plot.
		The number under each box-plot is the geometric mean running time.
		Solutions that violate the balance constraint are still included in the geometric mean running time.
		Instances exceeding the time limit are included as well, with running time set to the time limit.
		The top plot shows the sequential partitioners and \texttt{Mt-KaHyPar} on benchmark set A with a time limit of 8 hours, the bottom left shows \texttt{Mt-KaHyPar}, \texttt{Zoltan} and \texttt{PaToH-D} on benchmark set B with a time limit of 2 hours, and the bottom right plot shows \texttt{Mt-KaHyPar} and the parallel graph partitioners \texttt{Mt-KaHIP}, \texttt{Mt-Metis} and \texttt{ParHIP} on the graph benchmark set used in the \texttt{Mt-KaHIP} paper~\cite{MT-KAHIP} without a time limit.
	}
	\label{fig:appendix:absolute_times}
\end{figure*}

\clearpage

\section{Graph Partitioning and Quality with Different Number of Threads}\label{sec:appendix:graph_partitioning_and_thread_dependent_quality}

\begin{figure*}[!ht]
  \centering
  \begin{minipage}{.49\textwidth}
    \externalizedfigure{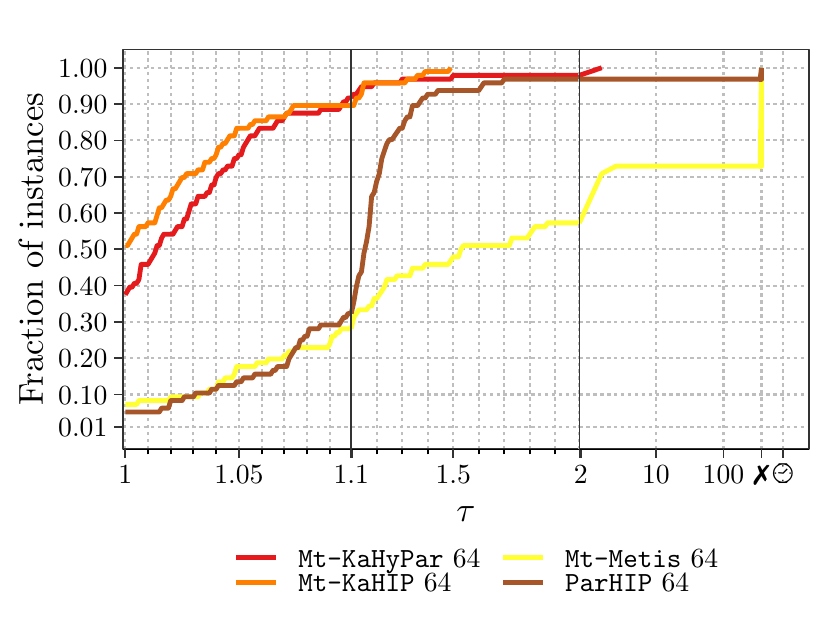}{experiments/graph/graph_quality.tex} %
    \vspace{-0.5cm}
  \end{minipage} %
  \begin{minipage}{.49\textwidth}
    \externalizedfigure{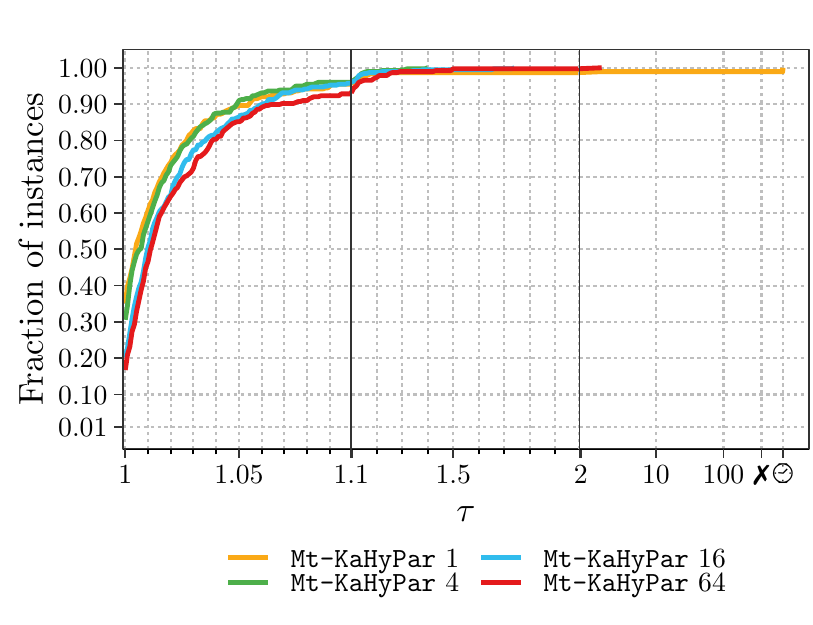}{experiments/quality/quality.tex} %
    \vspace{-0.5cm}
  \end{minipage} %
  \caption{Performance profiles comparing the quality of different parallel graph partitioners ($p = 64$) with \ExpAlgo{Mt-KaHyPar}{64} on the benchmark set used for \texttt{Mt-KaHIP} (left) and comparing the quality of \mtkahypar~with increasing number of threads on benchmark set B (right).}
  \label{fig:appendix:graph_partitioning_and_thread_dependent_quality}
  \vspace{-0.5cm}
\end{figure*}


\section{Effectiveness Tests}\label{sec:appendix:effectiveness}

\begin{figure*}[!ht]
  \centering
  \begin{minipage}{.3\textwidth}
    \vspace{0.185cm}
    \hspace{-0.75cm}
    \externalizedfigure{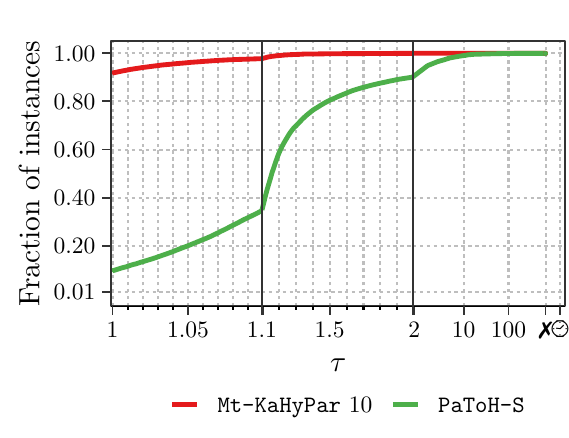}{experiments/kahypar_benchmark_set/mt_kahypar_vs_patoh_s.tex} %
    \vspace{-0.5cm}
  \end{minipage} %
  \begin{minipage}{.3\textwidth}
    \vspace{0.185cm}
    \externalizedfigure{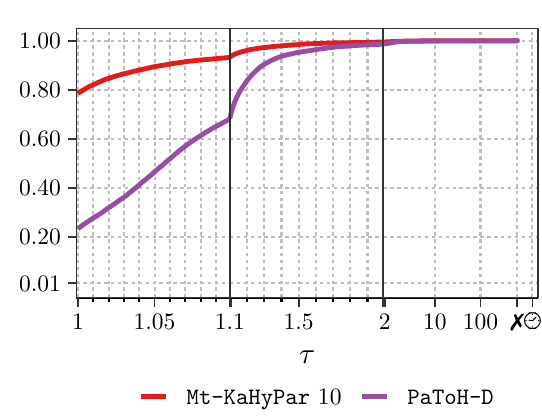}{experiments/kahypar_benchmark_set/mt_kahypar_vs_patoh_d.tex} %
    \vspace{-0.5cm}
  \end{minipage} %
  \begin{minipage}{.3\textwidth}
    \hspace{0.25cm}
    \vspace{-0.185cm}
    \externalizedfigure{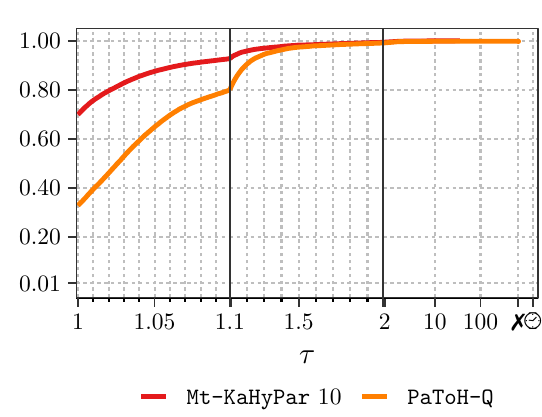}{experiments/kahypar_benchmark_set/mt_kahypar_vs_patoh_q.tex} %
    \vspace{-0.5cm}
  \end{minipage} %

  \begin{minipage}{.3\textwidth}
    \vspace{0.185cm}
    \hspace{-0.75cm}
    \externalizedfigure{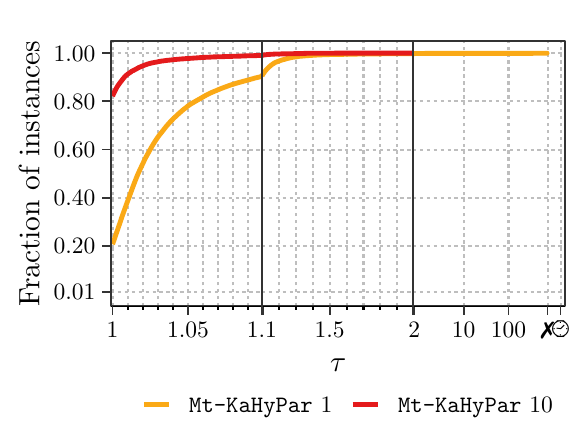}{experiments/kahypar_benchmark_set/mt_kahypar_1_vs_mt_kahypar_10.tex} %
    \vspace{-0.5cm}
  \end{minipage} %
  \begin{minipage}{.3\textwidth}
    \vspace{0.185cm}
    \externalizedfigure{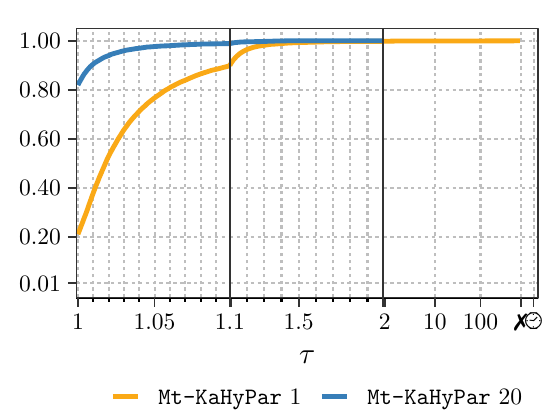}{experiments/kahypar_benchmark_set/mt_kahypar_1_vs_mt_kahypar_20.tex} %
    \vspace{-0.5cm}
  \end{minipage} %
  \begin{minipage}{.3\textwidth}
    \hspace{0.25cm}
    \vspace{-0.185cm}
    \externalizedfigure{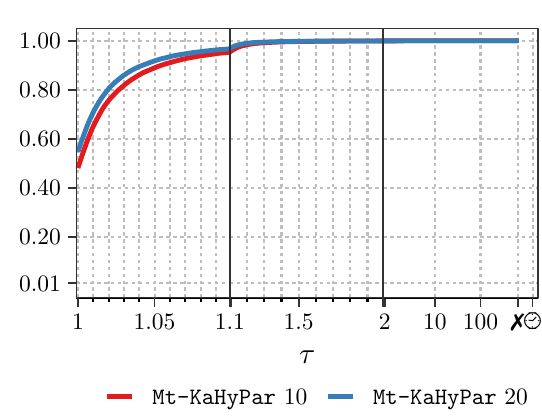}{experiments/kahypar_benchmark_set/mt_kahypar_10_vs_mt_kahypar_20.tex} %
    \vspace{-0.5cm}
  \end{minipage} %
  \vspace{-0.25cm}

  \caption{Effectiveness tests comparing \ExpAlgo{Mt-KaHyPar}{10} with \texttt{PaToH}, and itself with $1$ and $20$ threads on benchmark set A.
  	The faster algorithm performs (up to 10) repetitions so that both algorithms invest similar amounts of time.}\label{fig:appendix:effectiveness}
\end{figure*}

\twocolumn

\end{document}

%% file: macros.tex
\newcommand{\placeholder}[2]{\DTLfetch{#1}{key}{#2}{value}}
\newcommand{\externalizedfigure}[2]{
  \includegraphics{experiments/mt_kahypar-figure#1}
}

\newcommand{\mtkahypar}{\texttt{Mt-KaHyPar}}
\newcommand{\mtkahip}{\texttt{Mt-KaHIP}~\cite{MT-KAHIP}}
\newcommand{\mtmetis}{\texttt{Mt-Metis}~\cite{MT-METIS-REFINEMENT,MT-METIS}}

\newcommand{\parkway}{\texttt{Parkway}~\cite{PARKWAY-2}}
\newcommand{\zoltan}{\texttt{Zoltan}~\cite{ZOLTAN}}

\newcommand{\parmetis}{\texttt{ParMetis}~\cite{PARMETIS}}

\newcommand{\parhip}{\texttt{ParHIP}~\cite{PARHIP}}
\newcommand{\kahyparnocite}{\texttt{KaHyPar}}
\newcommand{\kahypar}{\kahyparnocite~\cite{KAHYPAR-K, KAHYPAR-MF, KAHYPAR-CA}}
\newcommand{\hmetis}{\texttt{hMetis}~\cite{HMETIS-K, HMETIS}}
\newcommand{\patoh}{\texttt{PaToH}~\cite{PATOH}}
\newcommand{\kahyparhfc}{\texttt{KaHyPar-HFC}~\cite{KAHYPAR-HFC}}
\newcommand{\kahyparca}{\texttt{KaHyPar-CA}~\cite{KAHYPAR-CA}}

\newcommand{\hmetisr}{\texttt{hMetis-R}}
\newcommand{\patohq}{\texttt{PaToH-Q}}
\newcommand{\patohd}{\texttt{PaToH-D}}
\newcommand{\patohs}{\texttt{PaToH-S}}

\newcommand{\Partition}{\ensuremath{\mathrm{\Pi}}}%
\newcommand{\incnets}{\ensuremath{\mathrm{I}}}%
\newcommand{\pinsinpart}{\ensuremath{\mathrm{\Phi}}}

\newcommand{\con}{\ensuremath{\lambda}}
\newcommand{\conset}{\ensuremath{\Lambda}}

\newcommand{\outmoves}{\mu_{\text{out}}}
\newcommand{\inmoves}{\mu_{\text{in}}}

\newcommand{\maxdeg}[1]{\ensuremath{\Delta_{#1}}}
\newcommand{\maxsize}[1]{\ensuremath{\Delta_{#1}}}

\newcommand{\meddeg}{\ensuremath{\widetilde{d(v)}}}
\newcommand{\medsize}{\ensuremath{\widetilde{|e|}}}

\newcommand{\cluster}{\mathcal{C}}

\newcommand{\SeqAlgo}[1]{\texttt{#1}}
\newcommand{\ExpAlgo}[2]{\texttt{#1}~$#2$}

\definecolor{fuchsiapink}{rgb}{1.0, 0.47, 1.0}

\definecolor{utahcrimson}{rgb}{0.83, 0.0, 0.25}

\newcommand{\plusplus}{\texttt{++}}

\newcommand{\gpp}[1]{g\plusplus#1}

%% file: sections/related_work.tex
We refer the reader to surveys~\cite{ALPERT-SURVEY,GRAPH-SURVEY,PAPA-SURVEY,KAHYPAR-DIS} for a more general overview of (hyper)graph partitioning.
There are a number of shared-memory graph partitioners~\cite{MT-KAHIP, PT-SCOTCH, MT-METIS, MT-DIBAP}, distributed graph partitioners~\cite{PARMETIS, MPI-DIBAP, PARHIP, JOSTLE}, and distributed hypergraph partitioners~\cite{ZOLTAN, SHP, PARKWAY-2}.
We are not aware of published shared-memory hypergraph partitioning algorithms, except one paper~\cite{PARALLEL-PATOH} on parallelizing the clustering and matching algorithms used in the coarsening phase of \patoh.

The matching algorithms~\cite{PT-SCOTCH, ZOLTAN, PARMETIS, MT-METIS, JOSTLE} and clustering algorithms~\cite{MT-KAHIP, PATOH, PARALLEL-PATOH, HMETIS} used by different partitioners in the coarsening phase are directly amenable to parallelization, often with only minor quality losses~\cite{PARALLEL-PATOH}.

Parallel initial partitioning is usually done by calls to sequential multilevel algorithms with different random seeds~\cite{MT-KAHIP, PT-SCOTCH, ZOLTAN, KAPPA, PARKWAY-2}.
The parallel graph partitioners \mtmetis~and \parmetis~use parallel recursive bipartitioning and split up threads for independent recursive calls.

For the refinement phase, sequential partitioners use Fiduccia-Mattheyses (FM)~\cite{FM} or Kernighan-Lin (KL)~\cite{KL} local search heuristics to improve partitions at each level.
FM repeatedly performs a vertex move (vertex pair swap for KL) with the highest gain (may be negative), and then returns the best observed solution.
Allowing moves with negative gain enables FM and KL to escape local minima.
To speed up local search, both are often only initialized with boundary vertices.

FM and KL have been deemed notoriously difficult to parallelize due to their serial move order~\cite{MT-METIS-REFINEMENT} and are P-complete~\cite{PHARD}.
The graph partitioners \texttt{Scotch}~\cite{PT-SCOTCH}, \texttt{Jostle}~\cite{JOSTLE} and \texttt{KaPPa}~\cite{KAPPA} perform sequential $2$-way FM refinement on independent block pairs, which yields limited parallelism.
\texttt{Parkway}~\cite{PARKWAY-2}, \mtkahip, and \texttt{ParHiP}~\cite{PARHIP} use \emph{size-constrained label propagation}~\cite{LABEL_PROPAGATION, PARHIP}.
Vertices are visited in parallel.
For each vertex, the move with highest positive gain (based on the current view) is immediately performed.
The \emph{greedy refinement}~\cite{MT-METIS} of \texttt{Mt-Metis} extends label propagation by inserting boundary vertices into thread-local priority queues and repeatedly performing the highest positive gain move on each thread.
Note that greedy refinement and label propagation do not perform negative gain moves.
The \emph{hill-scanning} algorithm~\cite{MT-METIS-REFINEMENT} of \texttt{Mt-Metis} improves the greedy refinement by escaping local minima, but only uses an approximate gain definition.
If the next move from a thread-local priority queue has negative gain, it attempts to find additional negative gain moves around this move, which yield overall positive gain if performed together.
This is similar to the localized FM of KaHIP~\cite{KAFFPA} but hill-scanning performs the set of moves as soon as it yields positive combined gain, instead of rolling back to the best seen solution like KaHIP.

To the best of our knowledge, \mtkahip~is the only graph partitioner with a parallel $k$-way FM variant, but even this algorithm has a sequential part that can become a bottleneck~\cite[p.~133]{YAROSLAV_DISS}.
It is based on the \emph{localized multi-try} FM of \texttt{KaHiP}~\cite{KAFFPA}.
The threads perform FM local searches that do not overlap on vertices.
Each thread initializes its search with a different boundary vertex, and gradually expands around it by claiming neighbors of moved vertices.
After the searches terminate, the move sequences of the threads are concatenated into one, for which gains are recomputed sequentially.
The best prefix of that sequence is applied.

\texttt{SHP}~\cite{SHP} is a non-multilevel distributed hypergraph partitioner which uses a modified objective function to reduce the effect of zero gain moves.
It performs distributed label propagation but first collects the number of performed moves between blocks and approves them in a probabilistic fashion.

%% file: sections/coarsening.tex
The purpose of the coarsening phase is to provide a sequence of structurally similar and successively smaller (coarser) hypergraphs $\langle H_0 = H, H_1, \dots, H_r \rangle$, to enable fast improvements in the refinement phase.
We obtain a coarser hypergraph $H_{i+1}$ from the finer $H_i$ by contracting a vertex clustering of $H_i$.
Vertices of the same cluster $\cluster$ are merged into a \emph{coarse} vertex with weight $\sum_{v \in \cluster} c(v)$.
The incident nets of the coarse vertex are obtained from the union of the nets of its constituents.
For each net $e$ of $H_i$, its pins are replaced by their corresponding coarse vertex in $H_{i+1}$.
Nets consisting of a single pin are discarded.
From a set of identical nets (containing the same pins in $H_{i+1}$) we keep one representative and aggregate their weights.

The coarsening process is repeated until the coarsest hypergraph $H_r$ has less than $160 \cdot k$ vertices, which is sufficiently small for initial partitioning.
Contracting clusterings instead of matchings -- as traditionally done in graph partitioning -- was already proposed for \texttt{PaToH} and \texttt{hMetis-K}~\cite{PATOH, HMETIS-K}.
These algorithms are amenable to parallelization and there already exists a parallel implementation with good speedups on a commodity machine~\cite{PARALLEL-PATOH}.
While the different clustering algorithms share similarities, our algorithm uses less locking than those proposed in Ref.~\cite{PARALLEL-PATOH} and employs explicit on-the-fly conflict resolution.
After describing the clustering algorithm, we outline the parallel contraction algorithm
and then describe an approach to enhance coarsening by restricting contractions to densely connected regions that was already proposed for KaHyPar~\cite{KAHYPAR-CA}.

\subsection{Clustering.}\label{sec:coarsening:clustering}

Initially each vertex is in its own cluster.
We iterate over the vertices in random order in parallel.
For each vertex $u$ we compute the best cluster $\cluster$ to join according to the \emph{heavy-edge} rating function~\cite{PATOH,KAHYPAR-CA,HMETIS}
\[r(u, \cluster) := \sum_{e \in I(u) \cap I(\cluster)} \frac{\omega{(e)}}{|e| - 1},\]
which prefers clusters that share a large number of heavy nets of small size with $u$.
Additionally, $\cluster$ must fulfill a weight constraint $c(\cluster) \leq \lceil \frac{c(V)}{160\cdot k} \rceil$ to ensure that a balanced partition exists.

Each vertex can be in one of three states: \emph{singleton}, currently \emph{joining} a cluster, or part of a \emph{multi-vertex cluster}.
We use compare-and-swap instructions to ensure consistency of vertex states and fetch-and-add instructions for cluster weights.
If vertex $u$ is in a singleton cluster, we find a cluster $\cluster$ for $u$ to join, and atomically set the state of $u$ to joining.
The join operation succeeds directly if $\cluster$ is a multi-vertex cluster.
Otherwise, if $\cluster$ is a singleton, we try to atomically change the state of the vertex $v$ in $\cluster$ to joining.
If that succeeds, $u$ joins $\cluster$ and both $u$ and $v$ are marked as part of a multi-vertex cluster.
In the case that $v$ is currently trying to join some cluster, we spin in a busy-waiting loop until the state of $v$ is updated to multi-vertex cluster by some other thread, and then join its new multi-vertex cluster.
In that loop, we check if $u$ is part of a cycle of vertices trying to join each other.
If so, the vertex with the smallest ID in the cycle gets to join its desired cluster, thus breaking the cycle.

We proceed to the contraction phase after one pass over the vertices or as soon as the current clustering reduces the number of vertices by more than a factor of $2.5$.
The weight constraint on coarse vertices can lead to passes in which only few vertices are contracted.
If a pass reduces the number of vertices by less than a factor of $1.01$, we still perform the remaining contraction and then proceed to initial partitioning, even if the $160 \cdot k$ vertex limit is not reached.

\subsection{Contraction.}

Contracting vertex clusters into a coarser hypergraph consists of several steps.
First, we remap cluster IDs to a consecutive range by computing a parallel prefix sum on an array that has a one at position $i$ if $i$ is a used cluster ID and zero otherwise.
Then, we accumulate the weights of vertices in each cluster using atomic fetch-and-add instructions.
Subsequently, we map the pins of each net to their cluster IDs, and remove duplicate pins by sorting the remapped pins of each net and keeping only the first occurrence of a vertex.
Furthermore, we remove nets with a single pin and replace multiple identical nets by a single net with aggregated weight.
Finally, we construct two adjacency arrays used to iterate over pins and incident nets, by computing parallel prefix sums over the sizes of the remaining nets (resp. degrees)..

For identical net detection we parallelize the \textsc{InrSrt} algorithm of Aykanat et al.~\cite{INR-Source, INR}.
It uses \emph{fingerprints} $\sum_{v \in e} v^2$ to eliminate unnecessary pairwise comparisons between nets. 
Nets with different fingerprints or different sizes cannot be identical.
We distribute the fingerprints and their associated nets to the threads using a hash function.
Each thread sorts the nets by their fingerprint and size and then performs pairwise comparisons on the subranges of potentially identical nets.
We aggregate the weights of identical nets at a representative and mark the others as invalid.
We skip comparisons with invalidated nets.

\subsection{Community Detection Enhancement.}

We enhance the coarsening process by prohibiting contractions between sparsely connected vertices of the hypergraph.
\kahyparnocite~\cite{KAHYPAR-CA} uses community detection on the bipartite graph representation of the hypergraph~\cite{BIPARTITE-GRAPH-2} to identify \emph{densely connected} vertex clusters.
Contractions are then only allowed between vertices in the same cluster.
This substantially improves the quality of both the initial and the final partition.
We use the parallel Louvain method (PLM) of Staudt and Meyerhenke~\cite{Louvain, PARALLEL-LOUVAIN} for modularity maximization, which is a widely used objective function for community detection~\cite{modularity.np, Newman04}.
Initially, each vertex is in its own cluster.
PLM iterates in parallel over the vertices.
A vertex is moved to a cluster in its neighborhood with the highest positive modularity improvement.
This is repeated for up to 5 rounds or until less than 1\% of the vertices were moved in a round.
If any vertex was moved, the process is repeated on the graph obtained from contracting the clusters.

%% file: sections/initial_partitioning.tex
We compute initial $k$-way partitions via multilevel recursive bipartitioning using our parallel coarsening and refinement.
Initial $2$-way partitions on the coarsest hypergraphs are computed with the same portfolio of
algorithms that is used in KaHyPar~\cite{KaHyPar-R}.
The portfolio contains 9 different sequential algorithms, including greedy hypergraph growing~\cite{PATOH}, random assignment,
hypergraph growing with label propagation, and alternating BFS~\cite{KAHYPAR-IP}.
Each is executed 20 times, which results in a lot of available task-based parallelism.
Additionally, each initial bipartition is refined using label propagation, and each thread performs sequential $2$-way FM refinement on the best initial bipartition that it worked on.
Subsequently, we select the best of the initial bipartitions for the refinement stage.

Other parallel partitioners either perform independent calls to sequential partitioners~\cite{MT-KAHIP, PT-SCOTCH, ZOLTAN, KAPPA, PARKWAY-2}, or perform recursive bipartitioning but divide the thread pool to work on independent recursive calls~\cite{PARMETIS, MT-METIS}.
We argue that employing work-stealing in coarsening and refinement, e.g., using TBB's task scheduler and its parallel primitives, is vital for the performance of initial partitioning via recursive bipartitioning.
Even coarse hypergraphs may contain many large nets~\cite{SHP}, unlike graphs, where the number of edges is bounded by~$|V|^2$. 
Additionally, the coarsening algorithm may be unable to contract the hypergraph down to the $\frac{c(v)}{160 \cdot k}$ contraction limit, as discussed at the end of Section~\ref{sec:coarsening:clustering}.
This may lead to unequally dense sub-hypergraphs in recursive calls and thus to load imbalance.

%% file: sections/refinement.tex
\section{Parallel Gain Calculation}\label{sec:gains}

Because vertices are moved simultaneously it is inherently difficult for parallel partitioners to compute exact gain values.
For example, the gain of a vertex move can change between its initial calculation and the actual execution of the move.
Similarly, two simultaneous moves can worsen the solution quality, even if each individual gain suggested an improvement~\cite{MT-METIS}.
Our refinement algorithms have different requirements for calculating and verifying gains, depending on whether vertices are immediately moved as they are explored (label propagation, see Section~\ref{sec:label-propagation}) or not (FM, Section~\ref{sec:fm}).
In the following, we describe a technique named \emph{attributed gains} to track the overall improvement and double-check the gain of a move (label propagation), a parallel \emph{gain cache} (FM), and a novel algorithm for \emph{recomputing exact gains} of a move sequence in parallel (FM).

\mtkahypar~optimizes the connectivity metric $\sum_{e \in E} (\con(e) -1 ) \: \omega(e)$.
The gain $g_i(u)$ of moving vertex $u$ to block $V_i$ can be written as
\vspace{-0.15cm}
\begin{align*}
g_i(u) = b(u) - p_i(u) \\
b(u) := \omega(\{ e \in I(u) \mid \pinsinpart(e, \Partition(u)) = 1 \}) \\
p_i(u) := \omega(\{ e \in I(u) \mid \pinsinpart(e, V_i) = 0\}).
\end{align*}

A straightforward approach to calculate $g_i(u)$ would be to iterate over all incident nets $e \in E$, increase the gain by $\omega(e)$ if $\pinsinpart(e, \Partition(u)) = 1$, and decrease it by $\omega(e)$ if $\pinsinpart(e, V_i) = 0$.

\subsection{Data Structures for Gain Calculation.}

In our partition data structure we store and maintain the pin counts $\pinsinpart(e,V_i)$ and connectivity sets $\conset(e)$ for each net $e$ and block $V_i$.
To save memory, we use a packed representation with $\lceil \log({\max_{e \in E} |e|}) \rceil$ bits per entry for the $\pinsinpart(e,V_i)$-values, and we use a bitset of size $k$ to store the connectivity set $\conset(e)$.
Using \emph{pop-count} and \emph{count-leading-zeroes} instructions, we can efficiently iterate over the connectivity sets.
We use a separate spin-lock for each net $e$ to synchronize writes to $\pinsinpart(e,V_i)$ and $\conset(e)$, since we cannot use atomic fetch-and-add instructions due to the packed representation.
Each thread holds at most one spin-lock at a time, and, in particular, we do not lock all incidents nets of a vertex before moving it.
Reads are not synchronized.

\subsection{Attributed Gains.}\label{sec:attributed_gain}

Since we cannot rely on the correctness of computed gains, we additionally compute an attributed gain for each move based on the synchronized writes to $\pinsinpart(e, V_i)$.
We attribute a connectivity reduction by $\omega(e)$ to the move that reduces $\pinsinpart(e, V_i)$ to zero and an increase by $\omega(e)$ for increasing it to one.
Since we only lock one incident net at a time (not all at once) when moving a vertex, this scheme may distribute the contributions to different threads.
Hence, there is still no guarantee on the correctness of one attributed gain.
However, the sum of the attributed gains equals the overall connectivity reduction.
We use attributed gains as a secondary check for our label propagation refinement algorithm and to correctly track the connectivity metric.

\subsection{Gain Caching.}\label{sec:gain_cache}

For our FM algorithm, we propose to use a gain cache~\cite{KAHYPAR-K, FromPQs} as we do not move vertices immediately after exploring them.
We use atomic fetch-and-add instructions to update the gains as vertices are moved.
Instead of storing $g_i(u)$, we store $b(u)$ and $p_i(u)$ separately for each vertex $u$, so that changes to $b(u)$ only require one update, instead of updates to $k$ gain values.
Gains can be looked up in constant time and we use $\mathcal{O}(|V| \cdot (k+1))$ memory in total.
Now, let vertex $u$ be moved from $V_a$ to $V_b$.
For each net $e \in I(u)$, we update $b(u)$ and $p_i(u)$ using atomic fetch-and-add instructions as follows:
\begin{enumerate}
\item If $\Phi(e, V_a) = 0$, then $\forall v \in e: p_a(v) = p_a(v) + \omega(e)$
\item If $\Phi(e, V_a) = 1$, then $\forall v \in e \cap V_a: b(v) = b(v) + \omega(e)$
\item If $\Phi(e, V_b) = 1$, then $\forall v \in e: p_b(v) = p_b(v) - \omega(e)$
\item If $\Phi(e, V_b) = 2$, then $\forall v \in e \cap V_b: b(v) = b(v) - \omega(e)$
\end{enumerate}

The update conditions are checked via the synchronized writes to $\Phi(e, V_a)$ and $\Phi(e, V_b)$.
Since we use only one memory location for $b(u)$, and not one for each block, the term can no longer be correctly updated after vertex $u$ is moved.
Our FM algorithm is organized in rounds in which each vertex can be moved at most once.
Therefore, we recalculate $b(u)$ for every moved vertex $u$ after each round.


\subsection{Parallel Gain Recalculation.}\label{sec:fm:gain_recalc}

Finally, we propose a parallel algorithm to recompute exact gains of vertex moves if they are supposed
to be performed in a given order, as is the case for FM.
This approach is of independent interest to graph partitioning, as \mtkahip~does this step sequentially.

Given a sequence of vertex moves $\mathcal{M} = \langle m_1, \dots, m_t \rangle$, we want to compute the exact gain of each move, as though they were performed sequentially in this order.
Let $V_i^j$ denote the vertices in block $V_i$ after the first $j-1$ moves are performed, i.e., $m_j$ is the next move to be performed.
Now, let $m_j$ move vertex $u_j$ from block $V_a$ to $V_b$.
Recall that the gain of $m_j$ is $\omega(\{ e \in I(u_j) \mid \Phi(e, V_a^j) = 1 \}) - \omega(\{ e \in I(u_j) \mid \Phi(e, V_b^j) = 0 \})$.

The following lemma yields equivalent conditions to $\Phi(e, V_a^j) = 1$ and $\Phi(e, V_b^j) = 0$ that are easier and more memory-efficient to compute in parallel than to compute all $\Phi(e, V_a^j)$ values.
The argument relies on the fact that each vertex can be moved at most once during an FM round.
We define the out-moves $\outmoves(e, V_i)$ of $V_i$ and $e \in E$ as the move indices $j$ such that $u_j \in e$ and $m_j$ moves $u_j$ out of $V_i$, i.e., those moves for which a pin of $e$ is moved out of $V_i$.
Analogously, we define the in-moves $\inmoves(e, V_i)$ such that $m_j$ moves $u_j$ into $V_i$.
Further, we define $\Psi(e, V_i) := \Phi(e, V_i^1) - |\outmoves(e, V_i)|$ as the number of pins of $e$ initially in $V_i$ that were removed from $V_i$ over the course of the FM round.
For simpler notation, we use $\max(\emptyset) = 0$ and $\min(\emptyset) = \infty$.

\begin{lemma}\label{lemma:gain_equivalence}
	Let $j \in \{1, \dots, t\}$, let $m_j$ move $u_j$ from $V_a$ to $V_b$, and let $e \in I(u)$.

	Then $\Phi(e, V_a^j) = 1$ if and only if ${\Psi(e, V_a) = 0}$, ${j = \max(\outmoves(e, V_a))}$, and ${\min(\inmoves(e, V_a)) > j}$.

	Similarly, $\Phi(e, V_b^j) = 0$ if and only if $\Psi(e, V_b) = 0$, $j = \min(\inmoves(e, V_b))$, and $\max(\outmoves(e, V_b)) < j$.
\end{lemma}
\begin{proof}
	We only prove the first statement, the second can be proven similarly.
	"$\Rightarrow$":
	If $\Phi(e, V_a^j) = 1$ then $u_j$ is the last vertex of $V_a^1$ remaining, as no vertex moved into $V_a$ this round can be moved out again.
	Hence, it follows that $j = \max(\outmoves(e, V_a))$ and thus also $\Phi(e, V_a^1) = |\outmoves(e, V_a)|$.
	Now, we prove $l := \min(\inmoves(e, V_a)) > j$.
	First, observe that $l \neq j$ since $m_j$ moves out of $V_a$.
	Assume for the sake of contradiction that $l < j$.
	Then $\{u_l, u_j\} \subset V_a^j$ and thus $\Phi(e, V_a^j) \geq 2$, a contradiction to $\Phi(e, V_a^j) = 1$.

	"$\Leftarrow$ ":
	$\Psi(e, V_a) = 0$ means that all vertices of $V_a^1$ were moved out of $V_a$ in the current FM round.
	If $j=\max(\outmoves(e, V_a))$ then $m_j$ is the last of these moves.
	Additionally, $j < \min(\inmoves(e, V_a))$ means that no vertex was moved into $V_a$ prior to $m_j$.
  Combining these yields $\Phi(e, V_a^{j+1}) = 0$, from which $\Phi(e, V_a^j) = 1$ follows, because $m_j$ moves out of $V_a$.

\end{proof}

We iterate over $\mathcal{M}$ in parallel, and over the incident nets $e \in I(u_j)$ sequentially.
In each step, we use early-exit compare-and-swap loops to update $\max(\outmoves(e, V_a))$ and $\min(\inmoves(e, V_b))$, and atomic fetch-and-add instructions for $\Psi(e, V_a)$.
In a second step, we again iterate over $\mathcal{M}$ in parallel and use the data from the first step to assign the exact gain to each move using Lemma~\ref{lemma:gain_equivalence}.

The compare-and-swap loops require $O(|E| \cdot k)$ memory.
For nets $e$ with only few pins, it is feasible to recompute the data from Lemma~\ref{lemma:gain_equivalence} at every moved pin.
Currently, our implementation does the recomputation if $|e|=2$, thus saving $O(k)$ memory on every such edge.
This is sufficient to make the memory overhead negligible on our benchmark sets.

\section{Label Propagation}\label{sec:label-propagation}
Our first refinement algorithm is a parallel version of the greedy refinement heuristic implemented in \texttt{hMetis-K}~\cite{HMETIS-K}, which is
also called \emph{size-constrained label propagation} in \parhip.
It bears resemblance to the clustering algorithms for coarsening and community detection.
In each round, we iterate in parallel over all boundary vertices that are \emph{eligible} to be moved.
In the first round, each boundary vertex is eligible.
For each eligible vertex $u$ we compute the gain of moving it to any of its neighboring blocks by iterating over its incident nets $e \in \incnets(u)$ and their connectivity sets $\Lambda(e)$.
Then, we move $u$ to the block with the highest positive gain that does not violate the balance constraint.
We also perform moves with zero gain that improve the balance.
After the move, we attribute a gain to the move based on the partition data structure updates, as described in Section~\ref{sec:attributed_gain}.
If the attributed gain is negative, we revert the move.
We rarely observed this in practice.
If $u$ is moved, then $u$ and its neighbors are made eligible to move in the next round.
Experiments indicate that a small number of rounds (we use $5$) suffice for convergence~\cite{HMETIS-K}.
We use atomic fetch-and-add instruction to update block weights to ensure that the balance constraint is never violated.

\section{Direct $k$-way FM}\label{sec:fm}
Our second refinement algorithm is a parallel version of the classical FM local search~\cite{FM}.
Sequential FM consists of two phases: \emph{finding a sequence of moves} by repeatedly performing a feasible move with highest gain followed by \emph{reverting moves back to the prefix with the highest gain} in that sequence.
We implement a relaxed version of the first phase that performs non-overlapping \emph{localized} FM searches on different threads to obtain a global sequence of moves.
The second phase consists of two steps: recalculate the exact gain of each move in that sequence (for which we propose a novel parallel algorithm in Section~\ref{sec:fm:gain_recalc}), followed by a combined parallel prefix sum and reduce operation to find the prefix in that sequence that provides the best gain and is balanced.
In the following, we describe the first phase in more detail.

\subsection{Localized Searches.}
An FM \emph{round} starts with inserting all boundary vertices into a globally shared task queue and randomly shuffling the queue.
We perform independent \emph{localized} FM local searches to find promising vertex moves.
Each search gradually expands around a constant number (we use 25) of starting vertices polled from the task queue.
The searches are non-overlapping, i.e., threads acquire exclusive ownership of vertices, although hyperedges can touch multiple searches.
Inserting a vertex $u$ into the local search of a thread entails acquiring ownership of $u$ by atomically setting a bit, looking up the gains for moving $u$ to any of the blocks, and inserting $u$ into a priority queue (PQ) with the highest gain as key.
Gains are stored and updated using a gain cache as explained in Section~\ref{sec:gain_cache}, that is shared among threads.
In particular, gain cache updates are visible to other threads, which yields adequately accurate gains at the time a vertex is moved.

We repeatedly extract and perform the feasible move with highest gain in the PQ.
After extracting a vertex from the PQ, we recompute the best move for that vertex using the gain cache.
If the gain of this move is worse than the key stored in the PQ, we reinsert the vertex with its new gain.
This way, we update the PQs to changes made by other threads.
After moving a vertex, we insert its neighbors into the local search or update their gain in the PQ if they were already inserted.
We skip vertices that were already moved or are owned by other threads.
The localized search finishes if the PQ becomes empty, the adaptive stopping rule of \kahyparnocite~\cite{KAHYPAR-K, ADAPTIVE-STOP-RULE} is triggered, or more than a certain number of threads unsuccessfully polled the task queue.
The last condition ensures a certain level of parallelism.
After each localized search, we release ownership of vertices that were not moved so that other threads can use them.
If in the meantime another thread polled a released vertex from the task queue, we reinsert it.
We repeatedly start localized searches until the task queue is empty.
Once the task queue is empty, we proceed to the second phase, where we recalculate the gains of the global move sequence (see Section~\ref{sec:fm:gain_recalc}) and keep the highest gain prefix of the global move sequence.

\subsection{Applying Moves.}
After each localized search, we keep the moves in the best prefix of the thread-local move sequence, and revert the remaining moves.
We use fetch-and-add instructions to assign global IDs to vertex moves as they are applied to the global partition, and store them in a vector that represents the global move sequence.
We use two different ways to apply moves to the global partition.
In the \emph{global} approach, moves are immediately applied as the searches find them.
This means moves are also reverted on the global partition and other threads temporarily see these.
In the \emph{local} approach, moves are first applied to a thread-local partition and only the best local prefix is applied to the global partition.
This requires thread-local hash tables for storing changes to the partition and gains.
The local approach yields better solution quality, but can incur a large memory and running time overhead on instances with many large nets, e.g., dual SAT instances~\cite{SAT14}.
By default we use \emph{local}.
After each localized search we check if the memory consumption of the hash tables exceeds a threshold.
If the threshold is exceeded, we switch to \emph{global}.
After each FM round we check if a time limit depending on the coarsening time is exceeded.
If so, we switch to \emph{global} and stop releasing vertices at the end of localized searches.

\section{Rebalancing}\label{sec:rebalancing}
While our refinement algorithms are guaranteed to produce balanced solutions, it turns out that intermediate balance violations improve solution quality.
The intuition behind this is that localized FM searches may each find good improvements but their combination is barely infeasible.
Hence, we relax the balance constraint for the second FM phase (the rollback step) by using $\varepsilon' = 1.25 \cdot \varepsilon$ instead of $\varepsilon$.
This value was obtained from preliminary experiments.
If the partition is still imbalanced on the finest level, we rebalance it using an approach that is similar to label propagation.
We iterate over the vertices in parallel.
If a vertex is in an overloaded block and can be moved with non-negative gain, we perform the move immediately.
Negative-gain moves are collected in thread-local priority queues.
If the partition is still imbalanced after performing all non-negative gain moves, we perform the moves with the smallest increase in connectivity from the priority queues, in a second step.

%% file: sections/details.tex
In the following, we outline engineering aspects that are necessary to enhance the performance of our system.

\subsection{Memory Allocation.}

Memory allocations are a significant bottleneck, which is why we implement a custom memory pool.
We estimate the memory needed by our data structures, categorized by the stage in the multilevel algorithm, and then allocate for the peak memory usage across all stages.
At the end of a stage, we pass the memory along to the next stage, and initialize data structures in parallel.
Further, we use the TBB scalable allocator~\cite{TBB} for concurrent memory allocations.

\subsection{Rating Aggregation.}
For community detection and coarsening, we aggregate ratings for neighboring clusters of each vertex.
This is traditionally done with a vector of size $|V|$.
Cache friendliness (on non-shared caches) is fundamental since the bandwidth to main memory is shared between threads.
We use fixed-capacity linear probing hash tables, and resort to the vector if the expected fill ratio exceeds $\frac{1}{3}$ of the capacity.
We conservatively estimate the fill grade as the vertex degree for community detection, and as the sum of the sizes of incident nets for coarsening.

\subsection{Random Shuffling.}
For community detection and coarsening, we visit vertices in a random order.
We use a block-shuffling scheme that does not guarantee uniform randomness but suffices for our purposes.
We divide the array into $2p$ equally-sized blocks.
Each thread swaps two random blocks and shuffles them.


\subsection{FM -- Boundary Vertices.}
We iterate in parallel over all vertices to collect boundary vertices in a shared task queue consisting of thread-local vectors.
To randomize the order in which vertices are polled from the task queue, each thread shuffles its local vector.
Additionally, we use a \texttt{tbb::concurrent\_queue} to reinsert vertices that were not moved at the end of a localized search.
We only reinsert vertices that were polled from the task queue by another thread (which obviously could not claim ownership).
To find starting vertices, each thread first polls its local vector, then the \texttt{tbb::concurrent\_queue}, and then tries to steal from other threads.

%% file: sections/experiments.tex


Our code is implemented in \texttt{C++17}\footnote{\mtkahypar~is available from \url{https://github.com/kahypar/mt-kahypar}}, parallelized with work-stealing using the TBB library~\cite{TBB}, and compiled using \gpp{9.2} with the flags \texttt{-O3 -mtune=native -march=native}.
For parallel partitioners we add a suffix to their name to indicate the number of threads used, e.g., \ExpAlgo{Mt-KaHyPar}{64} for 64 threads.
We omit the suffix for sequential partitioners.


\subsection{Instances.}

Our instances are derived from four sources encompassing three application domains: the ISPD98 VLSI Circuit Benchmark Suite~\cite{ISPD98}, the DAC 2012 Routability-Driven Placement Contest~\cite{DAC}, the SuiteSparse Matrix Collection~\cite{SPM}, and the 2014 SAT Competition~\cite{SAT14}.
We translate sparse matrices to hypergraphs using the row-net model~\cite{PATOH} and SAT instances to three different hypergraph representations: \emph{literal}, \emph{primal}, and \emph{dual}~\cite{MANN-PAPA14, PAPA-MARKOV} (for more details see also~\cite{KAHYPAR-CA}).
All hypergraphs have unit vertex and net weights.

For comparison with sequential partitioners, we use the established benchmark set of Heuer and Schlag~\cite{KAHYPAR-CA}, which contains $488$ hypergraphs (referred to as set A).
Since set A contains many small instances, we composed a new benchmark set (referred to as set B) consisting of $94$ much larger hypergraphs to evaluate speedups and to compare \mtkahypar~with parallel partitioners.
It consists of $42$ instances from the SuiteSparse Matrix Collection~\cite{SPM} that represents a randomly sampled subset of all available matrices with more than
$15$ million non-zero entries, all ten~\texttt{DAC} instances \cite{DAC} of set A, as well as the $24$ largest SAT instances of set A extended with
$18$ additional larger instances from~\cite{SAT14}.
\footnote{The benchmark sets and detailed statistics of their properties are publicly available from \url{http://algo2.iti.kit.edu/heuer/alenex21/}.}

\subsection{System.}

We use two different machines. 
Machines of type A are nodes of a cluster with Intel Xeon Gold 6230 processors ($2$ sockets with $20$ cores each) running at $2.1$ GHz with $96$GB RAM.
These are used for the comparison with sequential partitioners, as these experiments take roughly 3.5 CPU years overall.
Machine B is an AMD EPYC Rome 7702P ($1$ socket with $64$ cores) running at $2.0$--$3.35$ GHz with $1024$GB RAM.

On set A and machine A, we compare \mtkahypar~to the following sequential hypergraph partitioners: \kahyparca~($n$-level partitioner with similiar algorithmic components as \mtkahypar), the latest version \kahyparhfc (extends \texttt{KaHyPar-CA} with flow-based refinement), the recursive bisection version (\hmetisr) of \texttt{hMetis}~$2.0$~\cite{HMETIS}, as well as the speed (\patohs), default (\patohd), and quality preset (\patohq) of \texttt{PaToH} $3.3$~\cite{PATOH}.

On set B and machine B, we compare \mtkahypar~to the distributed partitioner \texttt{Zoltan} $3.83$~\cite{ZOLTAN}.
Unfortunately, we are not able to include \parkway. 
The \texttt{Parkway} version available on GitHub\footnote{\url{https://github.com/parkway-partitioner/parkway}} either hangs infinitely or crashes with a segmentation fault. 
Additionally, we asked the author of \texttt{PaToH} for a version with parallel coarsening~\cite{PARALLEL-PATOH} but were unable to obtain it.
Since only one phase is parallel, even the good speedups of factor 5-6 on 8 cores in that phase do not translate to overall speedups greater than a factor of 2.


\subsection{Methodology.}

Experiments on set A are executed with $k \in \{2,4,8,16,32,64,128\}$, $\varepsilon = 0.03$, ten different seeds and a time limit of eight hours.
Experiments on set B are executed with $k \in \{2,8,16,64\}$, $\varepsilon = 0.03$, five seeds and a time limit of two hours.
Each partitioner optimizes the connectivity metric which we also refer to as the quality of a partition.

For each instance (hypergraph and number of blocks $k$), we aggregate quality and running times using the arithmetic
mean (over all seeds).
To further aggregate over multiple instances, we use the harmonic mean for relative speedups, and the geometric mean for absolute running times, to give each instance a comparable influence.
Runs with imbalanced partitions are not excluded from aggregated running times.
For runs that exceeded the time limit we use the time limit in the aggregates.
Only if all runs of an algorithm produced imbalanced partitions on an instance, we mark it with \ding{55} in the plots, and similarly mark
instances for which all runs exceeded the time limit with \ClockLogo.

To compare the solution quality of different algorithms, we use \emph{performance profiles}~\cite{PERFORMANCE-PROFILES}.
Let $\mathcal{A}$ be the set of all algorithms we want to compare, $\mathcal{I}$ the set of instances, and $q_{A}(I)$ the quality of algorithm
$A \in \mathcal{A}$ on instance $I \in \mathcal{I}$.
For each algorithm $A$, we plot the fraction of instances ($y$-axis) for which $q_A(I) \leq \tau \cdot \min_{A' \in \mathcal{A}}q_{A'}(I)$, where $\tau$ is on the $x$-axis.
For $\tau = 1$, the $y$-value indicates the percentage of instances for which an algorithm $A \in \mathcal{A}$ performs best.
Note that these plots relate the quality of an algorithm to the best solution and thus do not permit a full ranking of three or more algorithms.

Additionally, we use the \emph{effectiveness tests} proposed by Ahkremtsev et. al.~\cite{MT-KAHIP} to compare solution quality when two algorithms are given a similar running time, by performing virtual repetitions with the faster algorithm.
We generate virtual instances that we compare using performance profiles.
Consider two algorithms $A$ and $B$, and an instance $I$.
We first sample one run of both algorithms on $I$.
Let $t_A^1, t_B^1$ be their running times and assume that $t_A^1 \geq t_B^1$.
We sample additional runs without replacement for $B$ until their accumulated time exceeds $t_A^1$ or all runs have been sampled.
Let $t_B^2, \dots, t_B^l$ denote their running times.
We accept the last run with probability $(t_A^1 - \sum_{i = 1}^{l-1} t_B^i) / t_B^l$ so that the expected time for the sampled runs of $B$ equals $t_A^1$.
The solution quality is the minimum out of the sampled runs.
For each instance, we generate $20$ virtual instances.

\subsection{Scalability.}

\begin{filecontents*}{experiments/speedup/speedup.dat}
speedUpTotalTimeP4S0 = 3.4
speedUpTotalTimeP4S10 = 3.4
speedUpTotalTimeP4S100 = 3.4
speedUpTotalTimeP4S1000 = 3.4
speedUpTotalTimeP4S10000 = 3.1
speedUpTotalTimeP16S0 = 10.8
speedUpTotalTimeP16S10 = 10.9
speedUpTotalTimeP16S100 = 11.1
speedUpTotalTimeP16S1000 = 11.3
speedUpTotalTimeP16S10000 = 11.1
speedUpTotalTimeP64S0 = 18.4
speedUpTotalTimeP64S10 = 18.9
speedUpTotalTimeP64S100 = 23.5
speedUpTotalTimeP64S1000 = 27
speedUpTotalTimeP64S10000 = 24.4
speedUpCommunityDetectionP4S0 = 3.2
speedUpCommunityDetectionP4S10 = 3.1
speedUpCommunityDetectionP4S100 = 3
speedUpCommunityDetectionP4S1000 = 3.5
speedUpCommunityDetectionP4S10000 = NaN
speedUpCommunityDetectionP16S0 = 9.7
speedUpCommunityDetectionP16S10 = 9.5
speedUpCommunityDetectionP16S100 = 10.6
speedUpCommunityDetectionP16S1000 = 13.1
speedUpCommunityDetectionP16S10000 = NaN
speedUpCommunityDetectionP64S0 = 18.7
speedUpCommunityDetectionP64S10 = 19.1
speedUpCommunityDetectionP64S100 = 27.3
speedUpCommunityDetectionP64S1000 = 36.6
speedUpCommunityDetectionP64S10000 = NaN
speedUpCoarseningP4S0 = 3.6
speedUpCoarseningP4S10 = 3.5
speedUpCoarseningP4S100 = 3.3
speedUpCoarseningP4S1000 = 3.1
speedUpCoarseningP4S10000 = NaN
speedUpCoarseningP16S0 = 11.1
speedUpCoarseningP16S10 = 11.3
speedUpCoarseningP16S100 = 11.4
speedUpCoarseningP16S1000 = 11
speedUpCoarseningP16S10000 = NaN
speedUpCoarseningP64S0 = 22
speedUpCoarseningP64S10 = 24
speedUpCoarseningP64S100 = 26.5
speedUpCoarseningP64S1000 = 19.9
speedUpCoarseningP64S10000 = NaN
speedUpInitialPartitioningP4S0 = 3.7
speedUpInitialPartitioningP4S10 = 3.6
speedUpInitialPartitioningP4S100 = 3.4
speedUpInitialPartitioningP4S1000 = 3.4
speedUpInitialPartitioningP4S10000 = 3.1
speedUpInitialPartitioningP16S0 = 10.3
speedUpInitialPartitioningP16S10 = 13
speedUpInitialPartitioningP16S100 = 13.3
speedUpInitialPartitioningP16S1000 = 13.5
speedUpInitialPartitioningP16S10000 = 11.5
speedUpInitialPartitioningP64S0 = 7.4
speedUpInitialPartitioningP64S10 = 19.5
speedUpInitialPartitioningP64S100 = 30.4
speedUpInitialPartitioningP64S1000 = 35.7
speedUpInitialPartitioningP64S10000 = 30.3
speedUpLabelPropagationP4S0 = 3.4
speedUpLabelPropagationP4S10 = 3.3
speedUpLabelPropagationP4S100 = 3.8
speedUpLabelPropagationP4S1000 = NaN
speedUpLabelPropagationP4S10000 = NaN
speedUpLabelPropagationP16S0 = 6.3
speedUpLabelPropagationP16S10 = 9.8
speedUpLabelPropagationP16S100 = 12
speedUpLabelPropagationP16S1000 = NaN
speedUpLabelPropagationP16S10000 = NaN
speedUpLabelPropagationP64S0 = 6.3
speedUpLabelPropagationP64S10 = 13.4
speedUpLabelPropagationP64S100 = 24
speedUpLabelPropagationP64S1000 = NaN
speedUpLabelPropagationP64S10000 = NaN
speedUpFMP4S0 = 3.4
speedUpFMP4S10 = 3.5
speedUpFMP4S100 = 3.6
speedUpFMP4S1000 = 3.4
speedUpFMP4S10000 = 3
speedUpFMP16S0 = 9
speedUpFMP16S10 = 11
speedUpFMP16S100 = 10.9
speedUpFMP16S1000 = 12.5
speedUpFMP16S10000 = 10.5
speedUpFMP64S0 = 8.1
speedUpFMP64S10 = 21.9
speedUpFMP64S100 = 26.8
speedUpFMP64S1000 = 24.2
speedUpFMP64S10000 = 23.9
percentageInstancesGreaterTotalTimeS0 = 100
percentageInstancesGreaterTotalTimeS10 = 96.3
percentageInstancesGreaterTotalTimeS100 = 47.6
percentageInstancesGreaterTotalTimeS1000 = 14.4
percentageInstancesGreaterTotalTimeS10000 = 0.5
percentageInstancesGreaterCommunityDetectionS0 = 100
percentageInstancesGreaterCommunityDetectionS10 = 53.2
percentageInstancesGreaterCommunityDetectionS100 = 13.8
percentageInstancesGreaterCommunityDetectionS1000 = 1.1
percentageInstancesGreaterCommunityDetectionS10000 = 0
percentageInstancesGreaterCoarseningS0 = 100
percentageInstancesGreaterCoarseningS10 = 67.8
percentageInstancesGreaterCoarseningS100 = 17
percentageInstancesGreaterCoarseningS1000 = 1.1
percentageInstancesGreaterCoarseningS10000 = 0
percentageInstancesGreaterInitialPartitioningS0 = 100
percentageInstancesGreaterInitialPartitioningS10 = 48.4
percentageInstancesGreaterInitialPartitioningS100 = 19.4
percentageInstancesGreaterInitialPartitioningS1000 = 4
percentageInstancesGreaterInitialPartitioningS10000 = 0.3
percentageInstancesGreaterLabelPropagationS0 = 100
percentageInstancesGreaterLabelPropagationS10 = 24.2
percentageInstancesGreaterLabelPropagationS100 = 1.3
percentageInstancesGreaterLabelPropagationS1000 = 0
percentageInstancesGreaterLabelPropagationS10000 = 0
percentageInstancesGreaterFMS0 = 100
percentageInstancesGreaterFMS10 = 61.4
percentageInstancesGreaterFMS100 = 27.9
percentageInstancesGreaterFMS1000 = 7.2
percentageInstancesGreaterFMS10000 = 0.3
percentageInstancesLessTotalTimeS0 = 0
percentageInstancesLessTotalTimeS10 = 3.7
percentageInstancesLessTotalTimeS100 = 52.4
percentageInstancesLessTotalTimeS1000 = 85.6
percentageInstancesLessTotalTimeS10000 = 99.5
percentageInstancesLessCommunityDetectionS0 = 0
percentageInstancesLessCommunityDetectionS10 = 46.8
percentageInstancesLessCommunityDetectionS100 = 86.2
percentageInstancesLessCommunityDetectionS1000 = 98.9
percentageInstancesLessCommunityDetectionS10000 = 100
percentageInstancesLessCoarseningS0 = 0
percentageInstancesLessCoarseningS10 = 32.2
percentageInstancesLessCoarseningS100 = 83
percentageInstancesLessCoarseningS1000 = 98.9
percentageInstancesLessCoarseningS10000 = 100
percentageInstancesLessInitialPartitioningS0 = 0
percentageInstancesLessInitialPartitioningS10 = 51.6
percentageInstancesLessInitialPartitioningS100 = 80.6
percentageInstancesLessInitialPartitioningS1000 = 96
percentageInstancesLessInitialPartitioningS10000 = 99.7
percentageInstancesLessLabelPropagationS0 = 0
percentageInstancesLessLabelPropagationS10 = 75.8
percentageInstancesLessLabelPropagationS100 = 98.7
percentageInstancesLessLabelPropagationS1000 = 100
percentageInstancesLessLabelPropagationS10000 = 100
percentageInstancesLessFMS0 = 0
percentageInstancesLessFMS10 = 38.6
percentageInstancesLessFMS100 = 72.1
percentageInstancesLessFMS1000 = 92.8
percentageInstancesLessFMS10000 = 99.7
minSpeedUpOfBest1PercentP4 = 4.1
minSpeedUpOfBest5PercentP4 = 3.9
minSpeedUpOfBest10PercentP4 = 3.8
minSpeedUpOfBest20PercentP4 = 3.7
minSpeedUpOfBest25PercentP4 = 3.7
minSpeedUpOfBest30PercentP4 = 3.6
minSpeedUpOfBest40PercentP4 = 3.6
minSpeedUpOfBest50PercentP4 = 3.6
minSpeedUpOfBest60PercentP4 = 3.5
minSpeedUpOfBest70PercentP4 = 3.4
minSpeedUpOfBest75PercentP4 = 3.3
minSpeedUpOfBest80PercentP4 = 3.2
minSpeedUpOfBest90PercentP4 = 3
minSpeedUpOfBest95PercentP4 = 2.9
minSpeedUpOfBest99PercentP4 = 2.6
minSpeedUpOfBest1PercentP16 = 16.2
minSpeedUpOfBest5PercentP16 = 14.5
minSpeedUpOfBest10PercentP16 = 13.8
minSpeedUpOfBest20PercentP16 = 12.8
minSpeedUpOfBest25PercentP16 = 12.5
minSpeedUpOfBest30PercentP16 = 12.2
minSpeedUpOfBest40PercentP16 = 11.8
minSpeedUpOfBest50PercentP16 = 11.2
minSpeedUpOfBest60PercentP16 = 10.7
minSpeedUpOfBest70PercentP16 = 10.2
minSpeedUpOfBest75PercentP16 = 10
minSpeedUpOfBest80PercentP16 = 9.7
minSpeedUpOfBest90PercentP16 = 9
minSpeedUpOfBest95PercentP16 = 8.4
minSpeedUpOfBest99PercentP16 = 5.8
minSpeedUpOfBest1PercentP64 = 42.2
minSpeedUpOfBest5PercentP64 = 35.2
minSpeedUpOfBest10PercentP64 = 30.7
minSpeedUpOfBest20PercentP64 = 26.9
minSpeedUpOfBest25PercentP64 = 25.5
minSpeedUpOfBest30PercentP64 = 24.4
minSpeedUpOfBest40PercentP64 = 21.6
minSpeedUpOfBest50PercentP64 = 20
minSpeedUpOfBest60PercentP64 = 17.7
minSpeedUpOfBest70PercentP64 = 15.5
minSpeedUpOfBest75PercentP64 = 15
minSpeedUpOfBest80PercentP64 = 14.3
minSpeedUpOfBest90PercentP64 = 13.2
minSpeedUpOfBest95PercentP64 = 10.8
minSpeedUpOfBest99PercentP64 = 7.4
\end{filecontents*}
\DTLloaddb[noheader, keys={key,value}]{speedup}{experiments/speedup/speedup.dat}

\begin{table}
  \footnotesize
  \captionsetup{font=footnotesize}
  \centering
  \caption{Harmonic mean speedups over all instances and instances with a single-threaded running time $\ge 100$s for
           the total partition time (T), community detection (CD), coarsening (C), initial partitioning (IP), label propagation (LP) and FM.}
  \label{tbl:hmean_speed_ups}
  \begin{tabular}{lr|rrrrrr}
  \multicolumn{2}{r|}{Num. Threads} & T & CD & C & IP & LP & FM \\
  \midrule
  \multirow{3}{*}{All}        & $4$ &  $\placeholder{speedup}{speedUpTotalTimeP4S0}$   & $\placeholder{speedup}{speedUpCommunityDetectionP4S0}$   &
                                           $\placeholder{speedup}{speedUpCoarseningP4S0}$  & $\placeholder{speedup}{speedUpInitialPartitioningP4S0}$ &
                                           $\placeholder{speedup}{speedUpLabelPropagationP4S0}$ & $\placeholder{speedup}{speedUpFMP4S0}$ \\
                              & $16$ & $\placeholder{speedup}{speedUpTotalTimeP16S0}$   & $\placeholder{speedup}{speedUpCommunityDetectionP16S0}$   &
                                           $\placeholder{speedup}{speedUpCoarseningP16S0}$  & $\placeholder{speedup}{speedUpInitialPartitioningP16S0}$ &
                                           $\placeholder{speedup}{speedUpLabelPropagationP16S0}$ & $\placeholder{speedup}{speedUpFMP16S0}$ \\
                              & $64$ & $\placeholder{speedup}{speedUpTotalTimeP64S0}$   & $\placeholder{speedup}{speedUpCommunityDetectionP64S0}$  &
                                           $\placeholder{speedup}{speedUpCoarseningP64S0}$  & $\placeholder{speedup}{speedUpInitialPartitioningP64S0}$ &
                                           $\placeholder{speedup}{speedUpLabelPropagationP64S0}$ & $\placeholder{speedup}{speedUpFMP64S0}$ \\
  \multirow{3}{*}{$\ge 100$s} & $4$  & $\placeholder{speedup}{speedUpTotalTimeP4S100}$ & $\placeholder{speedup}{speedUpCommunityDetectionP4S100}$ &
                                           $\placeholder{speedup}{speedUpCoarseningP4S100}$ & $\placeholder{speedup}{speedUpInitialPartitioningP4S100}$ &
                                           $\placeholder{speedup}{speedUpLabelPropagationP4S100}$ & $\placeholder{speedup}{speedUpFMP4S100}$ \\
                              & $16$ & $\placeholder{speedup}{speedUpTotalTimeP16S100}$ & $\placeholder{speedup}{speedUpCommunityDetectionP16S100}$ &
                                           $\placeholder{speedup}{speedUpCoarseningP16S100}$ & $\placeholder{speedup}{speedUpInitialPartitioningP16S100}$ &
                                           $\placeholder{speedup}{speedUpLabelPropagationP16S100}$ & $\placeholder{speedup}{speedUpFMP16S100}$ \\
                              & $64$ & $\placeholder{speedup}{speedUpTotalTimeP64S100}$ & $\placeholder{speedup}{speedUpCommunityDetectionP64S100}$ &
                                           $\placeholder{speedup}{speedUpCoarseningP64S100}$ & $\placeholder{speedup}{speedUpInitialPartitioningP64S100}$ &
                                           $\placeholder{speedup}{speedUpLabelPropagationP64S100}$ & $\placeholder{speedup}{speedUpFMP64S100}$ \\
  \midrule
  \multicolumn{2}{r|}{\footnotesize{Inst. $\ge 100$s $[\%]$}} & $\placeholder{speedup}{percentageInstancesGreaterTotalTimeS100}$ & $\placeholder{speedup}{percentageInstancesGreaterCommunityDetectionS100}$ &
                                                     $\placeholder{speedup}{percentageInstancesGreaterCoarseningS100}$ & $\placeholder{speedup}{percentageInstancesGreaterInitialPartitioningS100}$ &
                                                     $\placeholder{speedup}{percentageInstancesGreaterLabelPropagationS100}$ & $\placeholder{speedup}{percentageInstancesGreaterFMS100}$ \\
  \end{tabular}
\vspace*{-0.6cm}
\end{table}

\begin{figure*}
  \vspace*{-2cm}
  \centering
  \externalizedfigure{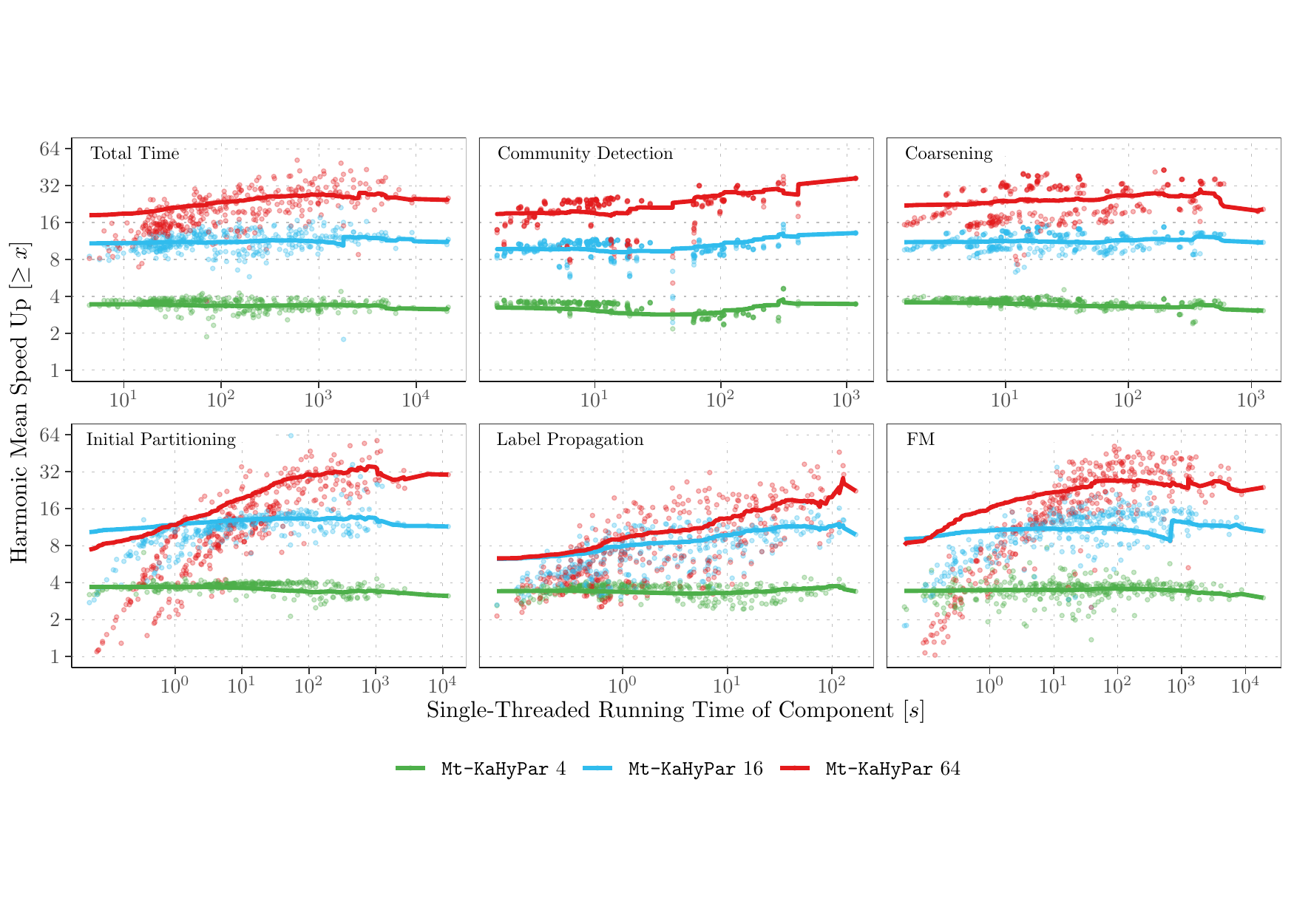}{experiments/speedup/speedup.tex} %
  \vspace*{-2.5cm}
  \caption{Arithmetic mean speedup per instance (points) and cumulative harmonic mean speedup (lines) of \mtkahypar~with $p \in \{4,16,64\}$ for each algorithmic component on set B. }
  \label{fig:speedup}
\end{figure*}

\begin{figure*}
  \centering
  \begin{minipage}{.3\textwidth}
    \hspace{-0.75cm}
    \externalizedfigure{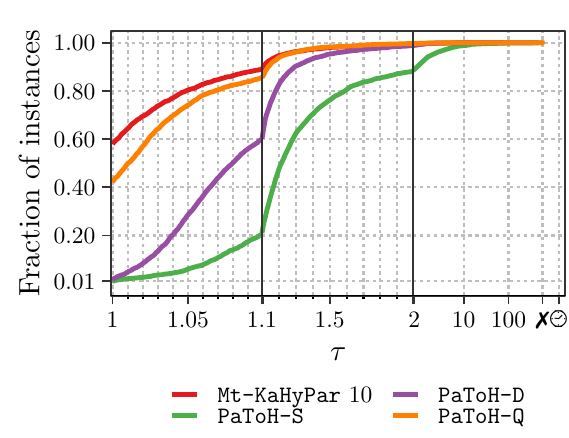}{experiments/kahypar_benchmark_set/kahypar_quality_patoh.tex} %
    \vspace{-0.5cm}
  \end{minipage} %
  \begin{minipage}{.3\textwidth}
    \externalizedfigure{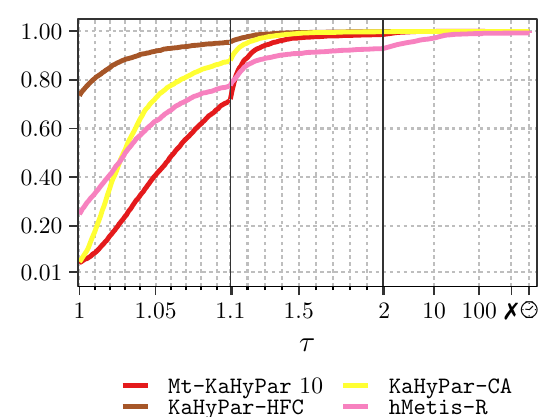}{experiments/kahypar_benchmark_set/kahypar_quality_high_quality.tex} %
    \vspace{-0.5cm}
  \end{minipage} %
  \begin{minipage}{.3\textwidth}
    \hspace{0.25cm}
    \externalizedfigure{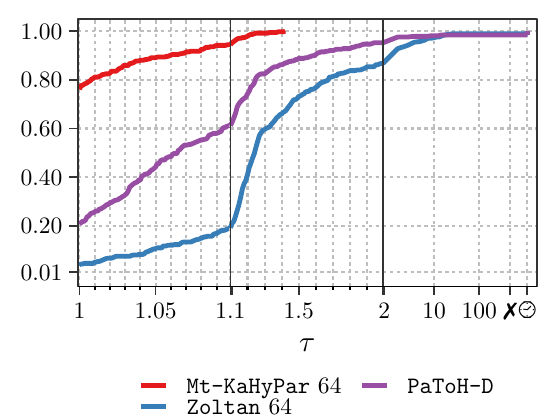}{experiments/big_benchmark_set/big_quality.tex} %
    \vspace{-0.5cm}
  \end{minipage} %
  \caption{Performance profiles comparing the solution quality of \mtkahypar~with \texttt{PaToH} (left),
           \texttt{hMetis} and \texttt{KaHyPar} (middle) on set A, as well as \texttt{Zoltan} and \texttt{PaToH} (right) on set B.}
  \label{fig:quality}
  \vspace{-0.25cm}
\end{figure*}

\begin{figure*}
  \hspace{0.65cm}
	\begin{minipage}{.3\textwidth}
		\vspace{0.185cm}
		\hspace{-0.75cm}
		\externalizedfigure{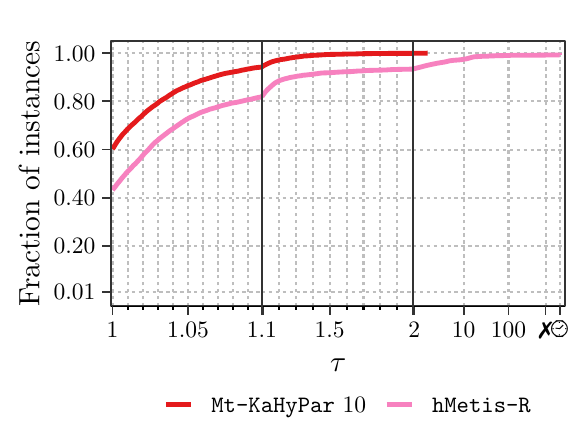}{experiments/kahypar_benchmark_set/mt_kahypar_vs_hmetis_r.tex} %
		\vspace{-0.5cm}
	\end{minipage} %
	\begin{minipage}{.3\textwidth}
		\vspace{0.185cm}
		\externalizedfigure{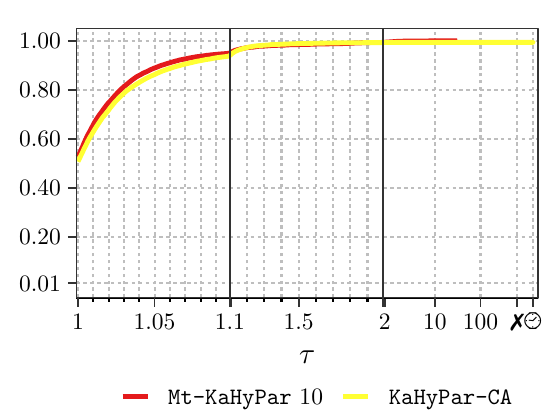}{experiments/kahypar_benchmark_set/mt_kahypar_vs_kahypar_ca.tex} %
		\vspace{-0.5cm}
	\end{minipage} %
	\begin{minipage}{.3\textwidth}
		\hspace{0.25cm}
		\vspace{-0.185cm}
		\externalizedfigure{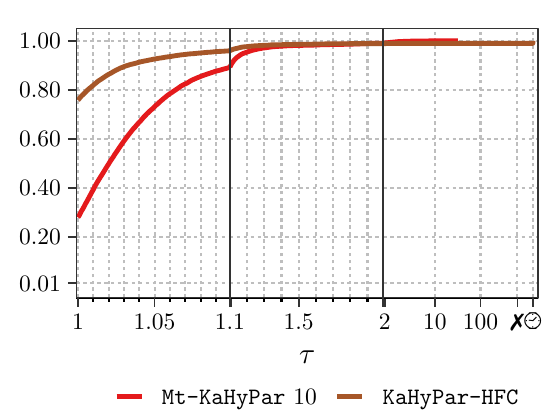}{experiments/kahypar_benchmark_set/mt_kahypar_vs_kahypar_hfc.tex} %
		\vspace{-0.5cm}
	\end{minipage} %
\caption{Effectiveness tests comparing \ExpAlgo{Mt-KaHyPar}{10} with \texttt{hMetis-R}, \texttt{KaHyPar-CA} and \texttt{KaHypar-HFC} on set A.
         The faster algorithm performs (up to 10) repetitions so that both algorithms invest similar amounts of time.}
\label{fig:effectiveness_tests_quality}
\end{figure*}

\begin{figure*}[!ht]
  \centering
  \begin{minipage}{.49\textwidth}
    \externalizedfigure{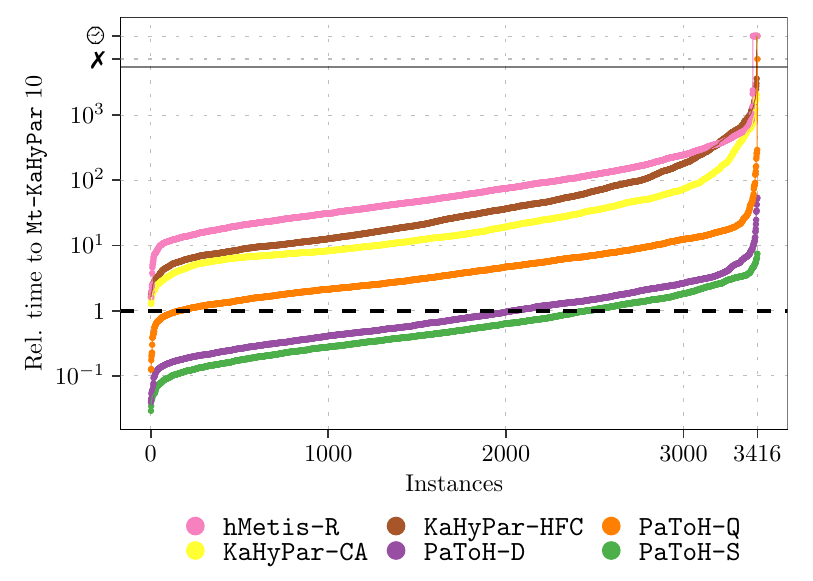}{experiments/kahypar_benchmark_set/kahypar_relative_runtime.tex} %
    \vspace{-0.25cm}
  \end{minipage} %
  \begin{minipage}{.49\textwidth}
    \externalizedfigure{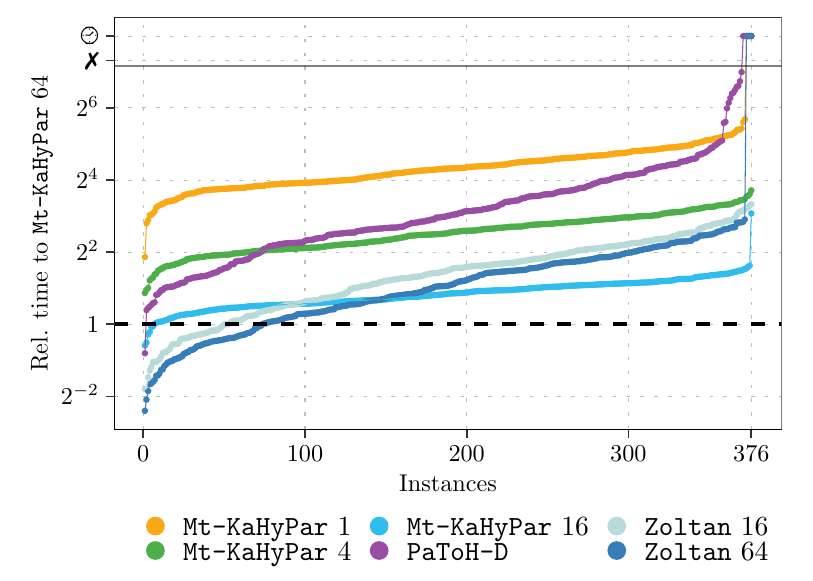}{experiments/big_benchmark_set/big_relative_runtime.tex} %
    \vspace{-0.25cm}
  \end{minipage} %
  \caption{Running times of different sequential resp. parallel partitioners relative
           to \ExpAlgo{Mt-KaHyPar}{10} on set A (left) resp. \ExpAlgo{Mt-KaHyPar}{64} on set B (right).}
  \label{fig:running_time}
  \vspace{-0.25cm}
\end{figure*}

In Figure~\ref{fig:speedup} and Table~\ref{tbl:hmean_speed_ups}, we summarize the speedups of \mtkahypar~with varying number of threads ($p \in \{4,16,64\}$) for each algorithmic component on set B.
In the plot, we represent the speedup of each instance as a point and the cumulative harmonic mean speedup over
all instances with a single-threaded running time $\ge x$ seconds with a line.

The overall harmonic mean speedup of \mtkahypar~is $\placeholder{speedup}{speedUpTotalTimeP4S0}$ for $p = 4$,
$\placeholder{speedup}{speedUpTotalTimeP16S0}$ for $p = 16$ and $\placeholder{speedup}{speedUpTotalTimeP64S0}$ for $p = 64$.
If we only consider instances with a single-threaded running time $\ge 100$s, we achieve a harmonic mean speedup
of $\placeholder{speedup}{speedUpTotalTimeP64S100}$ for $p = 64$.
For $p = 4$, the speedup is at least $\placeholder{speedup}{minSpeedUpOfBest90PercentP4}$ on over $90\%$ of our instances.

Community detection and coarsening share many similarities in their implementation and both show reliable speedups for increasing number of threads.
For the remaining three components, we observe that longer single-threaded execution leads to substantially better speedups.
For initial partitioning, increasing the number of threads from $16$ to $64$ can even be harmful for instances with a single-threaded running time of one second or less.
While label propagation refinement yields the least promising speedups, it is substantially faster than the other components, taking less than 10\% of the overall running time on over 95\% of the instances, for $p=64$.
Our algorithms do not perform any expensive arithmetic calculations.
Hence, scalability is limited by memory bandwidth, which makes achieving perfect speedups difficult.

\subsection{Comparison with other Systems.}

\begin{filecontents*}{experiments/kahypar_benchmark_set/sequential.dat}
mtKaHyParvsPaToHQMtKaHyPar10Tau100 = 58
mtKaHyParvsPaToHQMtKaHyPar10TauInv100 = 42
mtKaHyParvsPaToHQMtKaHyPar10Tau105 = 80.4
mtKaHyParvsPaToHQMtKaHyPar10TauInv105 = 19.6
mtKaHyParvsPaToHQMtKaHyPar10Tau110 = 89.1
mtKaHyParvsPaToHQMtKaHyPar10TauInv110 = 10.9
mtKaHyParvsPaToHQMtKaHyPar10Tau125 = 95.6
mtKaHyParvsPaToHQMtKaHyPar10TauInv125 = 4.4
mtKaHyParvsPaToHQMtKaHyPar10Tau150 = 98.1
mtKaHyParvsPaToHQMtKaHyPar10TauInv150 = 1.9
mtKaHyParvsPaToHQMtKaHyPar10Tau200 = 99.1
mtKaHyParvsPaToHQMtKaHyPar10TauInv200 = 0.9
mtKaHyParvsPaToHQPaToHQTau100 = 43.1
mtKaHyParvsPaToHQPaToHQTauInv100 = 56.9
mtKaHyParvsPaToHQPaToHQTau105 = 74.2
mtKaHyParvsPaToHQPaToHQTauInv105 = 25.8
mtKaHyParvsPaToHQPaToHQTau110 = 85.6
mtKaHyParvsPaToHQPaToHQTauInv110 = 14.4
mtKaHyParvsPaToHQPaToHQTau125 = 95.3
mtKaHyParvsPaToHQPaToHQTauInv125 = 4.7
mtKaHyParvsPaToHQPaToHQTau150 = 98.2
mtKaHyParvsPaToHQPaToHQTauInv150 = 1.8
mtKaHyParvsPaToHQPaToHQTau200 = 99.6
mtKaHyParvsPaToHQPaToHQTauInv200 = 0.4
mtKaHyParvsPaToHDMtKaHyPar10Tau100 = 81.4
mtKaHyParvsPaToHDMtKaHyPar10TauInv100 = 18.6
mtKaHyParvsPaToHDMtKaHyPar10Tau105 = 91.2
mtKaHyParvsPaToHDMtKaHyPar10TauInv105 = 8.8
mtKaHyParvsPaToHDMtKaHyPar10Tau110 = 94.3
mtKaHyParvsPaToHDMtKaHyPar10TauInv110 = 5.7
mtKaHyParvsPaToHDMtKaHyPar10Tau125 = 97.1
mtKaHyParvsPaToHDMtKaHyPar10TauInv125 = 2.9
mtKaHyParvsPaToHDMtKaHyPar10Tau150 = 98.6
mtKaHyParvsPaToHDMtKaHyPar10TauInv150 = 1.4
mtKaHyParvsPaToHDMtKaHyPar10Tau200 = 99.5
mtKaHyParvsPaToHDMtKaHyPar10TauInv200 = 0.5
mtKaHyParvsPaToHDPaToHDTau100 = 19.4
mtKaHyParvsPaToHDPaToHDTauInv100 = 80.6
mtKaHyParvsPaToHDPaToHDTau105 = 41.6
mtKaHyParvsPaToHDPaToHDTauInv105 = 58.4
mtKaHyParvsPaToHDPaToHDTau110 = 66.4
mtKaHyParvsPaToHDPaToHDTauInv110 = 33.6
mtKaHyParvsPaToHDPaToHDTau125 = 88.7
mtKaHyParvsPaToHDPaToHDTauInv125 = 11.3
mtKaHyParvsPaToHDPaToHDTau150 = 95.7
mtKaHyParvsPaToHDPaToHDTauInv150 = 4.3
mtKaHyParvsPaToHDPaToHDTau200 = 98.9
mtKaHyParvsPaToHDPaToHDTauInv200 = 1.1
mtKaHyParvsPaToHSMtKaHyPar10Tau100 = 95.2
mtKaHyParvsPaToHSMtKaHyPar10TauInv100 = 4.8
mtKaHyParvsPaToHSMtKaHyPar10Tau105 = 97.8
mtKaHyParvsPaToHSMtKaHyPar10TauInv105 = 2.2
mtKaHyParvsPaToHSMtKaHyPar10Tau110 = 98.8
mtKaHyParvsPaToHSMtKaHyPar10TauInv110 = 1.2
mtKaHyParvsPaToHSMtKaHyPar10Tau125 = 99.6
mtKaHyParvsPaToHSMtKaHyPar10TauInv125 = 0.4
mtKaHyParvsPaToHSMtKaHyPar10Tau150 = 99.8
mtKaHyParvsPaToHSMtKaHyPar10TauInv150 = 0.2
mtKaHyParvsPaToHSMtKaHyPar10Tau200 = 100
mtKaHyParvsPaToHSMtKaHyPar10TauInv200 = 0
mtKaHyParvsPaToHSPaToHSTau100 = 5.4
mtKaHyParvsPaToHSPaToHSTauInv100 = 94.6
mtKaHyParvsPaToHSPaToHSTau105 = 13.7
mtKaHyParvsPaToHSPaToHSTauInv105 = 86.3
mtKaHyParvsPaToHSPaToHSTau110 = 30.6
mtKaHyParvsPaToHSPaToHSTauInv110 = 69.4
mtKaHyParvsPaToHSPaToHSTau125 = 60.7
mtKaHyParvsPaToHSPaToHSTauInv125 = 39.3
mtKaHyParvsPaToHSPaToHSTau150 = 79
mtKaHyParvsPaToHSPaToHSTauInv150 = 21
mtKaHyParvsPaToHSPaToHSTau200 = 89.9
mtKaHyParvsPaToHSPaToHSTauInv200 = 10.1
mtKaHyParvshMetisRMtKaHyPar10Tau100 = 40.1
mtKaHyParvshMetisRMtKaHyPar10TauInv100 = 59.9
mtKaHyParvshMetisRMtKaHyPar10Tau105 = 69.3
mtKaHyParvshMetisRMtKaHyPar10TauInv105 = 30.7
mtKaHyParvshMetisRMtKaHyPar10Tau110 = 85.1
mtKaHyParvshMetisRMtKaHyPar10TauInv110 = 14.9
mtKaHyParvshMetisRMtKaHyPar10Tau125 = 95.8
mtKaHyParvshMetisRMtKaHyPar10TauInv125 = 4.2
mtKaHyParvshMetisRMtKaHyPar10Tau150 = 98.9
mtKaHyParvshMetisRMtKaHyPar10TauInv150 = 1.1
mtKaHyParvshMetisRMtKaHyPar10Tau200 = 99.8
mtKaHyParvshMetisRMtKaHyPar10TauInv200 = 0.2
mtKaHyParvshMetisRhMetisRTau100 = 61
mtKaHyParvshMetisRhMetisRTauInv100 = 39
mtKaHyParvshMetisRhMetisRTau105 = 81.2
mtKaHyParvshMetisRhMetisRTauInv105 = 18.8
mtKaHyParvshMetisRhMetisRTau110 = 86.3
mtKaHyParvshMetisRhMetisRTauInv110 = 13.7
mtKaHyParvshMetisRhMetisRTau125 = 90.8
mtKaHyParvshMetisRhMetisRTauInv125 = 9.2
mtKaHyParvshMetisRhMetisRTau150 = 92.4
mtKaHyParvshMetisRhMetisRTauInv150 = 7.6
mtKaHyParvshMetisRhMetisRTau200 = 93.8
mtKaHyParvshMetisRhMetisRTauInv200 = 6.2
mtKaHyParvsKaHyParCAMtKaHyPar10Tau100 = 18.1
mtKaHyParvsKaHyParCAMtKaHyPar10TauInv100 = 81.9
mtKaHyParvsKaHyParCAMtKaHyPar10Tau105 = 72.2
mtKaHyParvsKaHyParCAMtKaHyPar10TauInv105 = 27.8
mtKaHyParvsKaHyParCAMtKaHyPar10Tau110 = 90
mtKaHyParvsKaHyParCAMtKaHyPar10TauInv110 = 10
mtKaHyParvsKaHyParCAMtKaHyPar10Tau125 = 96.8
mtKaHyParvsKaHyParCAMtKaHyPar10TauInv125 = 3.2
mtKaHyParvsKaHyParCAMtKaHyPar10Tau150 = 98.1
mtKaHyParvsKaHyParCAMtKaHyPar10TauInv150 = 1.9
mtKaHyParvsKaHyParCAMtKaHyPar10Tau200 = 99
mtKaHyParvsKaHyParCAMtKaHyPar10TauInv200 = 1
mtKaHyParvsKaHyParCAKaHyParCATau100 = 83.1
mtKaHyParvsKaHyParCAKaHyParCATauInv100 = 16.9
mtKaHyParvsKaHyParCAKaHyParCATau105 = 97
mtKaHyParvsKaHyParCAKaHyParCATauInv105 = 3
mtKaHyParvsKaHyParCAKaHyParCATau110 = 98.4
mtKaHyParvsKaHyParCAKaHyParCATauInv110 = 1.6
mtKaHyParvsKaHyParCAKaHyParCATau125 = 99.4
mtKaHyParvsKaHyParCAKaHyParCATauInv125 = 0.6
mtKaHyParvsKaHyParCAKaHyParCATau150 = 99.7
mtKaHyParvsKaHyParCAKaHyParCATauInv150 = 0.3
mtKaHyParvsKaHyParCAKaHyParCATau200 = 99.8
mtKaHyParvsKaHyParCAKaHyParCATauInv200 = 0.2
mtKaHyParvsKaHyParHFCMtKaHyPar10Tau100 = 7.4
mtKaHyParvsKaHyParHFCMtKaHyPar10TauInv100 = 92.6
mtKaHyParvsKaHyParHFCMtKaHyPar10Tau105 = 46.7
mtKaHyParvsKaHyParHFCMtKaHyPar10TauInv105 = 53.3
mtKaHyParvsKaHyParHFCMtKaHyPar10Tau110 = 77.1
mtKaHyParvsKaHyParHFCMtKaHyPar10TauInv110 = 22.9
mtKaHyParvsKaHyParHFCMtKaHyPar10Tau125 = 94.6
mtKaHyParvsKaHyParHFCMtKaHyPar10TauInv125 = 5.4
mtKaHyParvsKaHyParHFCMtKaHyPar10Tau150 = 97.7
mtKaHyParvsKaHyParHFCMtKaHyPar10TauInv150 = 2.3
mtKaHyParvsKaHyParHFCMtKaHyPar10Tau200 = 98.7
mtKaHyParvsKaHyParHFCMtKaHyPar10TauInv200 = 1.3
mtKaHyParvsKaHyParHFCKaHyParHFCTau100 = 93.9
mtKaHyParvsKaHyParHFCKaHyParHFCTauInv100 = 6.1
mtKaHyParvsKaHyParHFCKaHyParHFCTau105 = 98
mtKaHyParvsKaHyParHFCKaHyParHFCTauInv105 = 2
mtKaHyParvsKaHyParHFCKaHyParHFCTau110 = 98.8
mtKaHyParvsKaHyParHFCKaHyParHFCTauInv110 = 1.2
mtKaHyParvsKaHyParHFCKaHyParHFCTau125 = 99.5
mtKaHyParvsKaHyParHFCKaHyParHFCTauInv125 = 0.5
mtKaHyParvsKaHyParHFCKaHyParHFCTau150 = 99.7
mtKaHyParvsKaHyParHFCKaHyParHFCTauInv150 = 0.3
mtKaHyParvsKaHyParHFCKaHyParHFCTau200 = 99.8
mtKaHyParvsKaHyParHFCKaHyParHFCTauInv200 = 0.2
mtKaHyParvsFastPartitionerMtKaHyPar10Tau100 = 57.9
mtKaHyParvsFastPartitionerMtKaHyPar10TauInv100 = 42.1
mtKaHyParvsFastPartitionerMtKaHyPar10Tau105 = 80.3
mtKaHyParvsFastPartitionerMtKaHyPar10TauInv105 = 19.7
mtKaHyParvsFastPartitionerMtKaHyPar10Tau110 = 89.1
mtKaHyParvsFastPartitionerMtKaHyPar10TauInv110 = 10.9
mtKaHyParvsFastPartitionerMtKaHyPar10Tau125 = 95.6
mtKaHyParvsFastPartitionerMtKaHyPar10TauInv125 = 4.4
mtKaHyParvsFastPartitionerMtKaHyPar10Tau150 = 98
mtKaHyParvsFastPartitionerMtKaHyPar10TauInv150 = 2
mtKaHyParvsFastPartitionerMtKaHyPar10Tau200 = 99.1
mtKaHyParvsFastPartitionerMtKaHyPar10TauInv200 = 0.9
mtKaHyParvsFastPartitionerPaToHSTau100 = 1
mtKaHyParvsFastPartitionerPaToHSTauInv100 = 99
mtKaHyParvsFastPartitionerPaToHSTau105 = 6
mtKaHyParvsFastPartitionerPaToHSTauInv105 = 94
mtKaHyParvsFastPartitionerPaToHSTau110 = 21
mtKaHyParvsFastPartitionerPaToHSTauInv110 = 79
mtKaHyParvsFastPartitionerPaToHSTau125 = 55.4
mtKaHyParvsFastPartitionerPaToHSTauInv125 = 44.6
mtKaHyParvsFastPartitionerPaToHSTau150 = 76.2
mtKaHyParvsFastPartitionerPaToHSTauInv150 = 23.8
mtKaHyParvsFastPartitionerPaToHSTau200 = 88.4
mtKaHyParvsFastPartitionerPaToHSTauInv200 = 11.6
mtKaHyParvsFastPartitionerPaToHDTau100 = 1.8
mtKaHyParvsFastPartitionerPaToHDTauInv100 = 98.2
mtKaHyParvsFastPartitionerPaToHDTau105 = 28.6
mtKaHyParvsFastPartitionerPaToHDTauInv105 = 71.4
mtKaHyParvsFastPartitionerPaToHDTau110 = 60.2
mtKaHyParvsFastPartitionerPaToHDTauInv110 = 39.8
mtKaHyParvsFastPartitionerPaToHDTau125 = 87.3
mtKaHyParvsFastPartitionerPaToHDTauInv125 = 12.7
mtKaHyParvsFastPartitionerPaToHDTau150 = 95.3
mtKaHyParvsFastPartitionerPaToHDTauInv150 = 4.7
mtKaHyParvsFastPartitionerPaToHDTau200 = 98.8
mtKaHyParvsFastPartitionerPaToHDTauInv200 = 1.2
mtKaHyParvsFastPartitionerPaToHQTau100 = 42.1
mtKaHyParvsFastPartitionerPaToHQTauInv100 = 57.9
mtKaHyParvsFastPartitionerPaToHQTau105 = 74.1
mtKaHyParvsFastPartitionerPaToHQTauInv105 = 25.9
mtKaHyParvsFastPartitionerPaToHQTau110 = 85.5
mtKaHyParvsFastPartitionerPaToHQTauInv110 = 14.5
mtKaHyParvsFastPartitionerPaToHQTau125 = 95.3
mtKaHyParvsFastPartitionerPaToHQTauInv125 = 4.7
mtKaHyParvsFastPartitionerPaToHQTau150 = 98.2
mtKaHyParvsFastPartitionerPaToHQTauInv150 = 1.8
mtKaHyParvsFastPartitionerPaToHQTau200 = 99.6
mtKaHyParvsFastPartitionerPaToHQTauInv200 = 0.4
mtKaHyParvsHighQualityPartitionerMtKaHyPar10Tau100 = 4.8
mtKaHyParvsHighQualityPartitionerMtKaHyPar10TauInv100 = 95.2
mtKaHyParvsHighQualityPartitionerMtKaHyPar10Tau105 = 40.8
mtKaHyParvsHighQualityPartitionerMtKaHyPar10TauInv105 = 59.2
mtKaHyParvsHighQualityPartitionerMtKaHyPar10Tau110 = 72
mtKaHyParvsHighQualityPartitionerMtKaHyPar10TauInv110 = 28
mtKaHyParvsHighQualityPartitionerMtKaHyPar10Tau125 = 92.5
mtKaHyParvsHighQualityPartitionerMtKaHyPar10TauInv125 = 7.5
mtKaHyParvsHighQualityPartitionerMtKaHyPar10Tau150 = 97.2
mtKaHyParvsHighQualityPartitionerMtKaHyPar10TauInv150 = 2.8
mtKaHyParvsHighQualityPartitionerMtKaHyPar10Tau200 = 98.6
mtKaHyParvsHighQualityPartitionerMtKaHyPar10TauInv200 = 1.4
mtKaHyParvsHighQualityPartitionerhMetisRTau100 = 24.8
mtKaHyParvsHighQualityPartitionerhMetisRTauInv100 = 75.2
mtKaHyParvsHighQualityPartitionerhMetisRTau105 = 63.1
mtKaHyParvsHighQualityPartitionerhMetisRTauInv105 = 36.9
mtKaHyParvsHighQualityPartitionerhMetisRTau110 = 77.7
mtKaHyParvsHighQualityPartitionerhMetisRTauInv110 = 22.3
mtKaHyParvsHighQualityPartitionerhMetisRTau125 = 87.8
mtKaHyParvsHighQualityPartitionerhMetisRTauInv125 = 12.2
mtKaHyParvsHighQualityPartitionerhMetisRTau150 = 90.9
mtKaHyParvsHighQualityPartitionerhMetisRTauInv150 = 9.1
mtKaHyParvsHighQualityPartitionerhMetisRTau200 = 92.8
mtKaHyParvsHighQualityPartitionerhMetisRTauInv200 = 7.2
mtKaHyParvsHighQualityPartitionerKaHyParCATau100 = 5.2
mtKaHyParvsHighQualityPartitionerKaHyParCATauInv100 = 94.8
mtKaHyParvsHighQualityPartitionerKaHyParCATau105 = 72.5
mtKaHyParvsHighQualityPartitionerKaHyParCATauInv105 = 27.5
mtKaHyParvsHighQualityPartitionerKaHyParCATau110 = 88
mtKaHyParvsHighQualityPartitionerKaHyParCATauInv110 = 12
mtKaHyParvsHighQualityPartitionerKaHyParCATau125 = 96.5
mtKaHyParvsHighQualityPartitionerKaHyParCATauInv125 = 3.5
mtKaHyParvsHighQualityPartitionerKaHyParCATau150 = 99.1
mtKaHyParvsHighQualityPartitionerKaHyParCATauInv150 = 0.9
mtKaHyParvsHighQualityPartitionerKaHyParCATau200 = 99.7
mtKaHyParvsHighQualityPartitionerKaHyParCATauInv200 = 0.3
mtKaHyParvsHighQualityPartitionerKaHyParHFCTau100 = 73.4
mtKaHyParvsHighQualityPartitionerKaHyParHFCTauInv100 = 26.6
mtKaHyParvsHighQualityPartitionerKaHyParHFCTau105 = 91.8
mtKaHyParvsHighQualityPartitionerKaHyParHFCTauInv105 = 8.2
mtKaHyParvsHighQualityPartitionerKaHyParHFCTau110 = 95.5
mtKaHyParvsHighQualityPartitionerKaHyParHFCTauInv110 = 4.5
mtKaHyParvsHighQualityPartitionerKaHyParHFCTau125 = 98.5
mtKaHyParvsHighQualityPartitionerKaHyParHFCTauInv125 = 1.5
mtKaHyParvsHighQualityPartitionerKaHyParHFCTau150 = 99.4
mtKaHyParvsHighQualityPartitionerKaHyParHFCTauInv150 = 0.6
mtKaHyParvsHighQualityPartitionerKaHyParHFCTau200 = 99.8
mtKaHyParvsHighQualityPartitionerKaHyParHFCTauInv200 = 0.2
mtKaHyParvsAllPartitionerMtKaHyPar10Tau100 = 4.3
mtKaHyParvsAllPartitionerMtKaHyPar10TauInv100 = 95.7
mtKaHyParvsAllPartitionerMtKaHyPar10Tau105 = 40.1
mtKaHyParvsAllPartitionerMtKaHyPar10TauInv105 = 59.9
mtKaHyParvsAllPartitionerMtKaHyPar10Tau110 = 71.3
mtKaHyParvsAllPartitionerMtKaHyPar10TauInv110 = 28.7
mtKaHyParvsAllPartitionerMtKaHyPar10Tau125 = 92
mtKaHyParvsAllPartitionerMtKaHyPar10TauInv125 = 8
mtKaHyParvsAllPartitionerMtKaHyPar10Tau150 = 97
mtKaHyParvsAllPartitionerMtKaHyPar10TauInv150 = 3
mtKaHyParvsAllPartitionerMtKaHyPar10Tau200 = 98.5
mtKaHyParvsAllPartitionerMtKaHyPar10TauInv200 = 1.5
mtKaHyParvsAllPartitionerhMetisRTau100 = 23.2
mtKaHyParvsAllPartitionerhMetisRTauInv100 = 76.8
mtKaHyParvsAllPartitionerhMetisRTau105 = 62.4
mtKaHyParvsAllPartitionerhMetisRTauInv105 = 37.6
mtKaHyParvsAllPartitionerhMetisRTau110 = 77.3
mtKaHyParvsAllPartitionerhMetisRTauInv110 = 22.7
mtKaHyParvsAllPartitionerhMetisRTau125 = 87.6
mtKaHyParvsAllPartitionerhMetisRTauInv125 = 12.4
mtKaHyParvsAllPartitionerhMetisRTau150 = 90.8
mtKaHyParvsAllPartitionerhMetisRTauInv150 = 9.2
mtKaHyParvsAllPartitionerhMetisRTau200 = 92.8
mtKaHyParvsAllPartitionerhMetisRTauInv200 = 7.2
mtKaHyParvsAllPartitionerKaHyParCATau100 = 4.9
mtKaHyParvsAllPartitionerKaHyParCATauInv100 = 95.1
mtKaHyParvsAllPartitionerKaHyParCATau105 = 70.7
mtKaHyParvsAllPartitionerKaHyParCATauInv105 = 29.3
mtKaHyParvsAllPartitionerKaHyParCATau110 = 87.3
mtKaHyParvsAllPartitionerKaHyParCATauInv110 = 12.7
mtKaHyParvsAllPartitionerKaHyParCATau125 = 96.1
mtKaHyParvsAllPartitionerKaHyParCATauInv125 = 3.9
mtKaHyParvsAllPartitionerKaHyParCATau150 = 98.9
mtKaHyParvsAllPartitionerKaHyParCATauInv150 = 1.1
mtKaHyParvsAllPartitionerKaHyParCATau200 = 99.7
mtKaHyParvsAllPartitionerKaHyParCATauInv200 = 0.3
mtKaHyParvsAllPartitionerKaHyParHFCTau100 = 70.8
mtKaHyParvsAllPartitionerKaHyParHFCTauInv100 = 29.2
mtKaHyParvsAllPartitionerKaHyParHFCTau105 = 90.4
mtKaHyParvsAllPartitionerKaHyParHFCTauInv105 = 9.6
mtKaHyParvsAllPartitionerKaHyParHFCTau110 = 94.9
mtKaHyParvsAllPartitionerKaHyParHFCTauInv110 = 5.1
mtKaHyParvsAllPartitionerKaHyParHFCTau125 = 98.2
mtKaHyParvsAllPartitionerKaHyParHFCTauInv125 = 1.8
mtKaHyParvsAllPartitionerKaHyParHFCTau150 = 99.4
mtKaHyParvsAllPartitionerKaHyParHFCTauInv150 = 0.6
mtKaHyParvsAllPartitionerKaHyParHFCTau200 = 99.7
mtKaHyParvsAllPartitionerKaHyParHFCTauInv200 = 0.3
mtKaHyParvsAllPartitionerPaToHSTau100 = 0.7
mtKaHyParvsAllPartitionerPaToHSTauInv100 = 99.3
mtKaHyParvsAllPartitionerPaToHSTau105 = 3
mtKaHyParvsAllPartitionerPaToHSTauInv105 = 97
mtKaHyParvsAllPartitionerPaToHSTau110 = 9.6
mtKaHyParvsAllPartitionerPaToHSTauInv110 = 90.4
mtKaHyParvsAllPartitionerPaToHSTau125 = 45
mtKaHyParvsAllPartitionerPaToHSTauInv125 = 55
mtKaHyParvsAllPartitionerPaToHSTau150 = 70.7
mtKaHyParvsAllPartitionerPaToHSTauInv150 = 29.3
mtKaHyParvsAllPartitionerPaToHSTau200 = 86.7
mtKaHyParvsAllPartitionerPaToHSTauInv200 = 13.3
mtKaHyParvsAllPartitionerPaToHDTau100 = 1.2
mtKaHyParvsAllPartitionerPaToHDTauInv100 = 98.8
mtKaHyParvsAllPartitionerPaToHDTau105 = 10.7
mtKaHyParvsAllPartitionerPaToHDTauInv105 = 89.3
mtKaHyParvsAllPartitionerPaToHDTau110 = 33.4
mtKaHyParvsAllPartitionerPaToHDTauInv110 = 66.6
mtKaHyParvsAllPartitionerPaToHDTau125 = 77.1
mtKaHyParvsAllPartitionerPaToHDTauInv125 = 22.9
mtKaHyParvsAllPartitionerPaToHDTau150 = 92.5
mtKaHyParvsAllPartitionerPaToHDTauInv150 = 7.5
mtKaHyParvsAllPartitionerPaToHDTau200 = 97.7
mtKaHyParvsAllPartitionerPaToHDTauInv200 = 2.3
mtKaHyParvsAllPartitionerPaToHQTau100 = 6.3
mtKaHyParvsAllPartitionerPaToHQTauInv100 = 93.7
mtKaHyParvsAllPartitionerPaToHQTau105 = 35.1
mtKaHyParvsAllPartitionerPaToHQTauInv105 = 64.9
mtKaHyParvsAllPartitionerPaToHQTau110 = 65
mtKaHyParvsAllPartitionerPaToHQTauInv110 = 35
mtKaHyParvsAllPartitionerPaToHQTau125 = 89.8
mtKaHyParvsAllPartitionerPaToHQTauInv125 = 10.2
mtKaHyParvsAllPartitionerPaToHQTau150 = 97.1
mtKaHyParvsAllPartitionerPaToHQTauInv150 = 2.9
mtKaHyParvsAllPartitionerPaToHQTau200 = 99
mtKaHyParvsAllPartitionerPaToHQTauInv200 = 1
MtKaHyParvsMtKaHyPar10GmeanTimeRatio = 6.5
MtKaHyParvsMtKaHyPar20GmeanTimeRatio = 9.1
MtKaHyParvshMetisRGmeanTimeRatio = 0.1
MtKaHyParvsKaHyParCAGmeanTimeRatio = 0.3
MtKaHyParvsKaHyParHFCGmeanTimeRatio = 0.2
MtKaHyParvsPaToHSGmeanTimeRatio = 12.6
MtKaHyParvsPaToHDGmeanTimeRatio = 8.4
MtKaHyParvsPaToHQGmeanTimeRatio = 1.7
MtKaHyPar10vsMtKaHyParGmeanTimeRatio = 0.2
MtKaHyPar10vsMtKaHyPar20GmeanTimeRatio = 1.4
MtKaHyPar10vshMetisRGmeanTimeRatio = 0
MtKaHyPar10vsKaHyParCAGmeanTimeRatio = 0.1
MtKaHyPar10vsKaHyParHFCGmeanTimeRatio = 0
MtKaHyPar10vsPaToHSGmeanTimeRatio = 1.9
MtKaHyPar10vsPaToHDGmeanTimeRatio = 1.3
MtKaHyPar10vsPaToHQGmeanTimeRatio = 0.3
MtKaHyPar20vsMtKaHyParGmeanTimeRatio = 0.1
MtKaHyPar20vsMtKaHyPar10GmeanTimeRatio = 0.7
MtKaHyPar20vshMetisRGmeanTimeRatio = 0
MtKaHyPar20vsKaHyParCAGmeanTimeRatio = 0
MtKaHyPar20vsKaHyParHFCGmeanTimeRatio = 0
MtKaHyPar20vsPaToHSGmeanTimeRatio = 1.4
MtKaHyPar20vsPaToHDGmeanTimeRatio = 0.9
MtKaHyPar20vsPaToHQGmeanTimeRatio = 0.2
hMetisRvsMtKaHyParGmeanTimeRatio = 9.5
hMetisRvsMtKaHyPar10GmeanTimeRatio = 62
hMetisRvsMtKaHyPar20GmeanTimeRatio = 86.3
hMetisRvsKaHyParCAGmeanTimeRatio = 3.3
hMetisRvsKaHyParHFCGmeanTimeRatio = 1.9
hMetisRvsPaToHSGmeanTimeRatio = 119.5
hMetisRvsPaToHDGmeanTimeRatio = 79.3
hMetisRvsPaToHQGmeanTimeRatio = 15.9
KaHyParCAvsMtKaHyParGmeanTimeRatio = 2.9
KaHyParCAvsMtKaHyPar10GmeanTimeRatio = 18.7
KaHyParCAvsMtKaHyPar20GmeanTimeRatio = 26.1
KaHyParCAvshMetisRGmeanTimeRatio = 0.3
KaHyParCAvsKaHyParHFCGmeanTimeRatio = 0.6
KaHyParCAvsPaToHSGmeanTimeRatio = 36.1
KaHyParCAvsPaToHDGmeanTimeRatio = 24
KaHyParCAvsPaToHQGmeanTimeRatio = 4.8
KaHyParHFCvsMtKaHyParGmeanTimeRatio = 5
KaHyParHFCvsMtKaHyPar10GmeanTimeRatio = 32.6
KaHyParHFCvsMtKaHyPar20GmeanTimeRatio = 45.3
KaHyParHFCvshMetisRGmeanTimeRatio = 0.5
KaHyParHFCvsKaHyParCAGmeanTimeRatio = 1.7
KaHyParHFCvsPaToHSGmeanTimeRatio = 62.8
KaHyParHFCvsPaToHDGmeanTimeRatio = 41.7
KaHyParHFCvsPaToHQGmeanTimeRatio = 8.4
PaToHSvsMtKaHyParGmeanTimeRatio = 0.1
PaToHSvsMtKaHyPar10GmeanTimeRatio = 0.5
PaToHSvsMtKaHyPar20GmeanTimeRatio = 0.7
PaToHSvshMetisRGmeanTimeRatio = 0
PaToHSvsKaHyParCAGmeanTimeRatio = 0
PaToHSvsKaHyParHFCGmeanTimeRatio = 0
PaToHSvsPaToHDGmeanTimeRatio = 0.7
PaToHSvsPaToHQGmeanTimeRatio = 0.1
PaToHDvsMtKaHyParGmeanTimeRatio = 0.1
PaToHDvsMtKaHyPar10GmeanTimeRatio = 0.8
PaToHDvsMtKaHyPar20GmeanTimeRatio = 1.1
PaToHDvshMetisRGmeanTimeRatio = 0
PaToHDvsKaHyParCAGmeanTimeRatio = 0
PaToHDvsKaHyParHFCGmeanTimeRatio = 0
PaToHDvsPaToHSGmeanTimeRatio = 1.5
PaToHDvsPaToHQGmeanTimeRatio = 0.2
PaToHQvsMtKaHyParGmeanTimeRatio = 0.6
PaToHQvsMtKaHyPar10GmeanTimeRatio = 3.9
PaToHQvsMtKaHyPar20GmeanTimeRatio = 5.4
PaToHQvshMetisRGmeanTimeRatio = 0.1
PaToHQvsKaHyParCAGmeanTimeRatio = 0.2
PaToHQvsKaHyParHFCGmeanTimeRatio = 0.1
PaToHQvsPaToHSGmeanTimeRatio = 7.5
PaToHQvsPaToHDGmeanTimeRatio = 5
FasterThanMtKaHyParOn = 99.9
SlowerThanMtKaHyParOn = 0.1
FasterThanMtKaHyParOn1PercentByFactor = 9
FasterThanMtKaHyParOn5PercentByFactor = 8.2
FasterThanMtKaHyParOn10PercentByFactor = 7.9
FasterThanMtKaHyParOn25PercentByFactor = 7.4
FasterThanMtKaHyParOn50PercentByFactor = 6.8
FasterThanMtKaHyParOn75PercentByFactor = 6.1
FasterThanMtKaHyParOn100PercentByFactor = 0.7
FasterThanMtKaHyPar20On = 3.5
SlowerThanMtKaHyPar20On = 96.5
FasterThanMtKaHyPar20On1PercentByFactor = 1.1
FasterThanMtKaHyPar20On5PercentByFactor = 1
FasterThanMtKaHyPar20On10PercentByFactor = 0.9
FasterThanMtKaHyPar20On25PercentByFactor = 0.8
FasterThanMtKaHyPar20On50PercentByFactor = 0.7
FasterThanMtKaHyPar20On75PercentByFactor = 0.6
FasterThanMtKaHyPar20On100PercentByFactor = 0.4
FasterThanhMetisROn = 100
SlowerThanhMetisROn = 0
FasterThanhMetisROn1PercentByFactor = 632
FasterThanhMetisROn5PercentByFactor = 390.1
FasterThanhMetisROn10PercentByFactor = 276.1
FasterThanhMetisROn25PercentByFactor = 130.3
FasterThanhMetisROn50PercentByFactor = 56.6
FasterThanhMetisROn75PercentByFactor = 27.6
FasterThanhMetisROn100PercentByFactor = 1.6
FasterThanKaHyParCAOn = 100
SlowerThanKaHyParCAOn = 0
FasterThanKaHyParCAOn1PercentByFactor = 642.8
FasterThanKaHyParCAOn5PercentByFactor = 188.2
FasterThanKaHyParCAOn10PercentByFactor = 88
FasterThanKaHyParCAOn25PercentByFactor = 37.5
FasterThanKaHyParCAOn50PercentByFactor = 14.1
FasterThanKaHyParCAOn75PercentByFactor = 7.7
FasterThanKaHyParCAOn100PercentByFactor = 1.3
FasterThanKaHyParHFCOn = 100
SlowerThanKaHyParHFCOn = 0
FasterThanKaHyParHFCOn1PercentByFactor = 1052.3
FasterThanKaHyParHFCOn5PercentByFactor = 465.3
FasterThanKaHyParHFCOn10PercentByFactor = 229.6
FasterThanKaHyParHFCOn25PercentByFactor = 74.9
FasterThanKaHyParHFCOn50PercentByFactor = 26.5
FasterThanKaHyParHFCOn75PercentByFactor = 11.3
FasterThanKaHyParHFCOn100PercentByFactor = 1.6
FasterThanPaToHSOn = 28.1
SlowerThanPaToHSOn = 71.9
FasterThanPaToHSOn1PercentByFactor = 4.2
FasterThanPaToHSOn5PercentByFactor = 2.9
FasterThanPaToHSOn10PercentByFactor = 2.1
FasterThanPaToHSOn25PercentByFactor = 1.1
FasterThanPaToHSOn50PercentByFactor = 0.5
FasterThanPaToHSOn75PercentByFactor = 0.2
FasterThanPaToHSOn100PercentByFactor = 0
FasterThanPaToHDOn = 40.1
SlowerThanPaToHDOn = 59.9
FasterThanPaToHDOn1PercentByFactor = 8.4
FasterThanPaToHDOn5PercentByFactor = 4.1
FasterThanPaToHDOn10PercentByFactor = 2.9
FasterThanPaToHDOn25PercentByFactor = 1.6
FasterThanPaToHDOn50PercentByFactor = 0.7
FasterThanPaToHDOn75PercentByFactor = 0.4
FasterThanPaToHDOn100PercentByFactor = 0
FasterThanPaToHQOn = 95.3
SlowerThanPaToHQOn = 4.7
FasterThanPaToHQOn1PercentByFactor = 43.5
FasterThanPaToHQOn5PercentByFactor = 17.5
FasterThanPaToHQOn10PercentByFactor = 13.4
FasterThanPaToHQOn25PercentByFactor = 7.5
FasterThanPaToHQOn50PercentByFactor = 3.6
FasterThanPaToHQOn75PercentByFactor = 1.9
FasterThanPaToHQOn100PercentByFactor = 0.1
\end{filecontents*}
\DTLloaddb[noheader, keys={key,value}]{sequential}{experiments/kahypar_benchmark_set/sequential.dat}
\begin{filecontents*}{experiments/big_benchmark_set/parallel.dat}
mtKaHyParvsZoltanMtKaHyParTau100 = 94.9
mtKaHyParvsZoltanMtKaHyParTauInv100 = 5.1
mtKaHyParvsZoltanMtKaHyParTau105 = 97.3
mtKaHyParvsZoltanMtKaHyParTauInv105 = 2.7
mtKaHyParvsZoltanMtKaHyParTau110 = 98.4
mtKaHyParvsZoltanMtKaHyParTauInv110 = 1.6
mtKaHyParvsZoltanZoltanTau100 = 5.6
mtKaHyParvsZoltanZoltanTauInv100 = 94.4
mtKaHyParvsZoltanZoltanTau105 = 12.2
mtKaHyParvsZoltanZoltanTauInv105 = 87.8
mtKaHyParvsZoltanZoltanTau110 = 22.3
mtKaHyParvsZoltanZoltanTauInv110 = 77.7
mtKaHyParvsZoltanZoltanTau125 = 55.6
mtKaHyParvsZoltanZoltanTauInv125 = 44.4
mtKaHyParvsZoltanZoltanTau150 = 73.9
mtKaHyParvsZoltanZoltanTauInv150 = 26.1
mtKaHyParvsZoltanZoltanTau200 = 86.7
mtKaHyParvsZoltanZoltanTauInv200 = 13.3
mtKaHyPar64vsPaToHMtKaHyParTau100 = 77.7
mtKaHyPar64vsPaToHMtKaHyParTauInv100 = 22.3
mtKaHyPar64vsPaToHMtKaHyParTau105 = 90.2
mtKaHyPar64vsPaToHMtKaHyParTauInv105 = 9.8
mtKaHyPar64vsPaToHMtKaHyParTau110 = 95.2
mtKaHyPar64vsPaToHMtKaHyParTauInv110 = 4.8
mtKaHyPar64vsPaToHMtKaHyParTau125 = 99.2
mtKaHyPar64vsPaToHMtKaHyParTauInv125 = 0.8
mtKaHyPar64vsPaToHPaToHDTau100 = 22.3
mtKaHyPar64vsPaToHPaToHDTauInv100 = 77.7
mtKaHyPar64vsPaToHPaToHDTau105 = 44.7
mtKaHyPar64vsPaToHPaToHDTauInv105 = 55.3
mtKaHyPar64vsPaToHPaToHDTau110 = 61.7
mtKaHyPar64vsPaToHPaToHDTauInv110 = 38.3
mtKaHyPar64vsPaToHPaToHDTau125 = 81.1
mtKaHyPar64vsPaToHPaToHDTauInv125 = 18.9
mtKaHyPar64vsPaToHPaToHDTau150 = 88.8
mtKaHyPar64vsPaToHPaToHDTauInv150 = 11.2
mtKaHyPar64vsPaToHPaToHDTau200 = 95.2
mtKaHyPar64vsPaToHPaToHDTauInv200 = 4.8
mtKaHyPar4vsPaToHMtKaHyPar4Tau100 = 83
mtKaHyPar4vsPaToHMtKaHyPar4TauInv100 = 17
mtKaHyPar4vsPaToHMtKaHyPar4Tau105 = 91.2
mtKaHyPar4vsPaToHMtKaHyPar4TauInv105 = 8.8
mtKaHyPar4vsPaToHMtKaHyPar4Tau110 = 96
mtKaHyPar4vsPaToHMtKaHyPar4TauInv110 = 4
mtKaHyPar4vsPaToHMtKaHyPar4Tau125 = 98.9
mtKaHyPar4vsPaToHMtKaHyPar4TauInv125 = 1.1
mtKaHyPar4vsPaToHPaToHDTau100 = 17
mtKaHyPar4vsPaToHPaToHDTauInv100 = 83
mtKaHyPar4vsPaToHPaToHDTau105 = 42.8
mtKaHyPar4vsPaToHPaToHDTauInv105 = 57.2
mtKaHyPar4vsPaToHPaToHDTau110 = 61.7
mtKaHyPar4vsPaToHPaToHDTauInv110 = 38.3
mtKaHyPar4vsPaToHPaToHDTau125 = 79.8
mtKaHyPar4vsPaToHPaToHDTauInv125 = 20.2
mtKaHyPar4vsPaToHPaToHDTau150 = 88.3
mtKaHyPar4vsPaToHPaToHDTauInv150 = 11.7
mtKaHyPar4vsPaToHPaToHDTau200 = 94.1
mtKaHyPar4vsPaToHPaToHDTauInv200 = 5.9
MtKaHyPar4vsMtKaHyParGmeanTimeRatio = 5.7
MtKaHyPar4vsZoltanGmeanTimeRatio = 2.8
MtKaHyPar4vsPaToHDGmeanTimeRatio = 0.7
MtKaHyParvsMtKaHyPar4GmeanTimeRatio = 0.2
MtKaHyParvsZoltanGmeanTimeRatio = 0.5
MtKaHyParvsPaToHDGmeanTimeRatio = 0.1
ZoltanvsMtKaHyPar4GmeanTimeRatio = 0.4
ZoltanvsMtKaHyParGmeanTimeRatio = 2.1
ZoltanvsPaToHDGmeanTimeRatio = 0.2
PaToHDvsMtKaHyPar4GmeanTimeRatio = 1.5
PaToHDvsMtKaHyParGmeanTimeRatio = 8.3
PaToHDvsZoltanGmeanTimeRatio = 4.1
FasterThanZoltanOn = 80.6
SlowerThanZoltanOn = 19.4
FasterThanZoltanOn1PercentByFactor = 207.2
FasterThanZoltanOn5PercentByFactor = 6.1
FasterThanZoltanOn10PercentByFactor = 5
FasterThanZoltanOn25PercentByFactor = 3.6
FasterThanZoltanOn50PercentByFactor = 2.1
FasterThanZoltanOn75PercentByFactor = 1.2
FasterThanZoltanOn100PercentByFactor = 0.2
FasterThanPaToHDOn = 99.7
SlowerThanPaToHDOn = 0.3
FasterThanPaToHDOn1PercentByFactor = 96.1
FasterThanPaToHDOn5PercentByFactor = 32.1
FasterThanPaToHDOn10PercentByFactor = 23
FasterThanPaToHDOn25PercentByFactor = 15.1
FasterThanPaToHDOn50PercentByFactor = 8.1
FasterThanPaToHDOn75PercentByFactor = 4.8
FasterThanPaToHDOn100PercentByFactor = 0.6
\end{filecontents*}
\DTLloaddb[noheader, keys={key,value}]{parallel}{experiments/big_benchmark_set/parallel.dat}
Figure~\ref{fig:quality} compares the quality of \mtkahypar~with different sequential (left and middle) and parallel hypergraph partitioners (right).
If we compare \ExpAlgo{Mt-KaHyPar}{10} with each sequential partitioner individually, \ExpAlgo{Mt-KaHyPar}{10} produces partitions with better quality than
\texttt{PaToH-S}, \texttt{PaToH-D}, \texttt{PaToH-Q}, \texttt{hMetis-R}, \texttt{KaHyPar-CA}, \texttt{KaHyPar-HFC} on
$\placeholder{sequential}{mtKaHyParvsPaToHSMtKaHyPar10Tau100}\%$, $\placeholder{sequential}{mtKaHyParvsPaToHDMtKaHyPar10Tau100}\%$,
$\placeholder{sequential}{mtKaHyParvsPaToHQMtKaHyPar10Tau100}\%$, $\placeholder{sequential}{mtKaHyParvshMetisRMtKaHyPar10Tau100}\%$,
$\placeholder{sequential}{mtKaHyParvsKaHyParCAMtKaHyPar10Tau100}\%$, and $\placeholder{sequential}{mtKaHyParvsKaHyParHFCMtKaHyPar10Tau100}\%$
of the instances of set A, respectively.
For set A, we use 10 cores with \mtkahypar~as this is a typical number of available cores in current commodity workstations.
On set B, \ExpAlgo{Mt-KaHyPar}{64} computes better partitions than \ExpAlgo{Zoltan}{64} on $\placeholder{parallel}{mtKaHyParvsZoltanMtKaHyParTau100}\%$ of the instances. Further, the quality of at least $\placeholder{parallel}{mtKaHyParvsZoltanZoltanTauInv110}\%$ of the partitions produced by \ExpAlgo{Zoltan}{64} are more than $10\%$ worse than those of \ExpAlgo{Mt-KaHyPar}{64}.

In the effectiveness tests presented in Figure~\ref{fig:effectiveness_tests_quality}, where algorithms receive similar running times via virtual repetitions, \ExpAlgo{Mt-KaHyPar}{10} finds better solutions than \texttt{hMetis-R} and similar quality solutions as \texttt{KaHyPar-CA}.
However, the more advanced approaches of \texttt{KaHyPar-HFC} cannot be compensated by additional repetitions of \mtkahypar.
Effectiveness tests with the other algorithms can be found in Figure~\ref{fig:appendix:effectiveness} in Appendix~\ref{sec:appendix:effectiveness}.

Figure~\ref{fig:running_time} shows running times of each partitioner relative to \mtkahypar.
We additionally report absolute running times as well as their geometric mean in Figure~\ref{fig:appendix:absolute_times} in Appendix~\ref{sec:appendix:absolute_times}.
On set A, \ExpAlgo{Mt-KaHyPar}{10} is consistently faster than \texttt{hMetis}, \texttt{KaHyPar} and \texttt{PaToH-Q}.
\texttt{PaToH-D} is faster on 60\% of the instances, and \texttt{PaToH-S} on 72\%, however, this is due to the small instances.
The geometric mean times of \texttt{PaToH-D} and \ExpAlgo{Mt-KaHyPar}{10} are similar at 1.17s vs 1.5s.
On the larger instances of set B, \ExpAlgo{Mt-KaHyPar}{4} is slightly faster \texttt{PaToH-D} while providing better solutions on 83\% of the instances, and being able to scalably utilize larger numbers of threads.
As shown in Figure~\ref{fig:appendix:graph_partitioning_and_thread_dependent_quality} (right) of Appendix~\ref{sec:appendix:graph_partitioning_and_thread_dependent_quality}, increasing the number of threads does not adversely affect solution quality of \texttt{Mt-KaHyPar}.
\ExpAlgo{Mt-KaHyPar}{64} dominates \ExpAlgo{Zoltan}{64} since it is 2.1 times faster in the geometric mean while achieving much better quality -- see Figure~\ref{fig:quality}.

Figure~\ref{fig:appendix:graph_partitioning_and_thread_dependent_quality} (left) in Appendix~\ref{sec:appendix:graph_partitioning_and_thread_dependent_quality} provides an additional comparison with parallel graph partitioners.
We see that \texttt{Mt-KaHyPar} achieves similar quality as \texttt{Mt-KaHiP} while only incurring a slowdown by a factor of $2$ in the geometric mean.
The slowdown is expected since we use substantially more complicated data structures.

%% file: references.bib
@String{Springer = {Springer-Verlag}}

@String{Springer = {Springer}}

@String{IEEE = {IEEE}}

@String{LIPIcs = {LIPIcs}}

@String{SIAM = {Society for Industrial and Applied Mathematics}}

@String{SIAM = {SIAM}}

@String{ESA11 = {19th European Symposium on Algorithms (ESA)}}

@String{SEA17 = {16th International Symposium on Experimental Algorithms (SEA)}}

@String{ISPD =  {International Symposium on Physical Design (ISPD)}}

@String{TKDE = {IEEE Transactions on Knowledge and Data Engineering}}

@inproceedings{ISPD98,
  author    = {C. J. Alpert},
  title     = {The ISPD98 Circuit Benchmark Suite},
  booktitle = ISPD,
  pages     = {80--85},
  year      = {1998},
  month     = {4}
}

@article{Newman04,
  author    = {M. E. J. Newman and
               M. Girvan},
  title     = {Finding and Evaluating Community Structure in Networks},
  journal   = {Physical Review},
  volume    = {69},
  issue     = {2},
  publisher = {American Physical Society},
  year      = {2004},
  month     = {2}
}

@article{Louvain,
  author    = {V. D. Blondel and
               J. Guillaume and
               R. Lambiotte and
               E. Lefebvre},
  title     = {Fast Unfolding of Communities in Large Networks},
  journal   = {Journal of Statistical Mechanics: Theory and Experiment},
  number    = {10},
  year      = {2008},
}

@inproceedings{kaffpa,
  author    = {P. Sanders and
               C. Schulz},
  title     = {{Engineering Multilevel Graph Partitioning Algorithms}},
  booktitle = ESA11,
  pages     = {469--480},
  publisher = {Springer},
  year      = {2011}
}

@article{modularity.np,
	author    = {U. Brandes and
               D. Delling and
               M. Gaertler and
               R. G\"orke and
               M. Hoefer and
               Z. Nikoloski and
               D. Wagner},
	title     = {On Modularity Clustering},
	journal   = TKDE,
	volume    = {20},
	number    = {2},
	pages     = {172--188},
	publisher = IEEE,
	year      = {2008}
}

@article{kappa,
  note = {-------------------UNUSED-------------------},
        title   = {{Engineering a Scalable High Quality Graph Partitioner}},
        author  = {Holtgrewe, M. and Sanders, P. and Schulz, C.},
        journal = {Proceedings of the 24th IEEE International Parallal and Distributed Processing Symposium},
        pages   = {1--12},
        year    = {2010}
}

@inproceedings{KaHyPar-CA,
 author    = {T. Heuer and
              S. Schlag},
 title     = {Improving Coarsening Schemes for Hypergraph Partitioning by Exploiting Community Structure},
 booktitle = SEA17,
 series =	{Leibniz International Proceedings in Informatics (LIPIcs)},
 pages     = {21:1--21:19},
 year      = {2017},
 month = {06},
 publisher =	{Schloss Dagstuhl -- Leibniz-Zentrum f{\"u}r Informatik}
}

@inproceedings{KaHyPar-K,
 author    = {Akhremtsev, Y. and
              Heuer, T. and
              Sanders, P. and
              Schlag, S.},
 title     = {Engineering a Direct \emph{k}-way Hypergraph Partitioning Algorithm},
 booktitle = {19th Workshop on Algorithm Engineering and Experiments (ALENEX)},
 pages     = {28--42},
 year      = {2017},
 publisher= {SIAM},
 month = {01}
}

@inproceedings{KaHyPar-R,
 author    = {Schlag, S. and
              Henne, V. and
              Heuer, T. and
              Meyerhenke, H. and
              Sanders, P. and
              Schulz, C.},
 title     = {$k$-way Hypergraph Partitioning via $n$-Level Recursive Bisection},
 booktitle = {18th Workshop on Algorithm Engineering and Experiments (ALENEX)},
 pages     = {53--67},
 year      = {2016},
 month = {01},
 publisher= {SIAM}
}
